\documentclass[12pt]{amsart}
\usepackage{xcolor}
\usepackage[all]{xy}
\usepackage{amsmath}
\usepackage{amsfonts}
\usepackage{amssymb}
\usepackage{amscd}
\usepackage{amsthm}
\usepackage{latexsym}
\usepackage{amsbsy}
\usepackage{graphicx}
\usepackage{epstopdf}
\usepackage{shuffle}

\input xypic

\usepackage{tikz-cd}

\usepackage{color}

\newtheorem{thm}{Theorem}[section]
\newtheorem{prop}[thm]{Proposition}
\newtheorem{cor}[thm]{Corollary}
\newtheorem{lem}[thm]{Lemma}

\newtheorem{defn}[thm]{Definition}
\newtheorem{rem}[thm]{Remark}
\newtheorem{ex}[thm]{Example}

\newtheorem{prob}[thm]{Problem}
\numberwithin{equation}{section}

\def\bC{{\mathbb C}}

\def\bP{{\mathbb P}}

\def\bR{{\mathbb R}}
\def\bS{{\mathbb S}}

\def\C{{\mathbb C}}

\def\R{{\mathbb R}}
\def\Z{{\mathbb Z}}

\def\cB{{\mathcal B}}

\def\cE{{\mathcal E}}
\def\cF{{\mathcal F}}
\def\cG{{\mathcal G}}

\def\cI{{\mathcal I}}
\def\cJ{{\mathcal J}}

\def\cL{{\mathcal L}}
\def\cM{{\mathcal M}}

\def\cO{{\mathcal O}}

\def\cR{{\mathcal R}}
\def\cS{{\mathcal S}}

\def\cU{{\mathcal U}}
\def\cV{{\mathcal V}}

\def\Hom{{\rm Hom}}

\def\fR{{\mathfrak R}}
\def\fV{{\mathfrak V}}

\title[Gabor frames and higher boundaries]{Gabor frames and higher dimensional boundaries in signal analysis on manifolds}
\author{Vasiliki Liontou, Matilde Marcolli}
\address{Department of Mathematics, University of Toronto,  ON  M5S 2E4, Canada}
\email{vasiliki.liontou@mail.utoronto.ca}
\address{Department of Mathematics and Department of Computing and Mathematical Sciences, 
California Institute of Technology, Pasadena, CA 91125, USA}
\email{matilde@caltech.edu}

\begin{document}
\maketitle

\begin{abstract}
We provide a construction of Gabor frames that encode local linearizations of a signal detected on a curved smooth manifold of arbitrary dimension, with Gabor filters that can detect the presence of higher-dimensional boundaries in the manifold signal. We describe an application in configuration spaces in robotics 
with sharp constrains.
The construction is a higher-dimensional generalization of the geometric setting 
developed for the study of signal analysis in the visual cortex. 
\end{abstract}

\section{Introduction}
\label{intro}

Gabor frames are an important tool for signal analysis, that allows for efficient
encoding and decoding of signals using filters that do not form an orthonormal basis
but can still perform some of the most important functions of classical Fourier analysis
and at the same time can optimize localization in both the position and frequency 
variables, minimizing uncertainty (for a general review of Gabor frames for signal
analysis, see \cite{Groch}). The construction of Gabor frames on a vector space $\R^n$
relies on the choice of a window function $\phi \in L^2(\R^n)$, usually assumed to be
of rapid decay, such as a multivariate Gaussian, and a lattice $\Lambda\subset \R^{2n}$
that provides the translation and modulation operators (time-frequency shifts) $\pi_\lambda=M_b T_a$,
for $\lambda=(a,b)\in \R^n\times\R^n$, that act on the window function by
$$ \pi_\lambda \phi(t)=M_b T_a  \phi (t)=e^{-2\pi i t\cdot b} \phi(t-a) \, . $$
The Gabor system $\cG(\Lambda,\phi)$ consists of the collection of $L^2(\R^n)$-functions (Gabor filters)
$$ \cG(\Lambda,\phi)=\{ \pi_\lambda \phi \,|\, \lambda\in \Lambda \}\, . $$
A Gabor system constructed in this way satisfies the frame condition if there are constants $C,C'>0$ such
that, for all $f\in L^2(\R^n)$,
$$ C\, \| f \|^2_{L^2(\R^n)} \leq \sum_{\lambda\in \Lambda} | \langle f , 
\pi_\lambda \phi \rangle |^2 \leq C' \, \| f \|^2_{L^2(\R^n)} \, . $$
In general, whether a Gabor system satisfies the frame condition is a very subtle property
that depends crucially on the lattice, and also depends on the choice of the test function.
In the case of Gabor systems in one-dimension with a Gaussian window function and
lattices of the form $\Lambda= \alpha \Z +\beta \Z$ the frame condition 
can be completely characterized as in  \cite{Lyu}, \cite{Sei} as the set 
$\{ (\alpha,\beta)\in \R_+^2 \,|\, \alpha\beta<1\}$. This result, for the same class of lattices, 
was extended to a much larger class of window functions in  \cite{GroSto}. In general,
the fundamental question of Gabor analysis, to identify, for a given window function
$\phi$ the set of lattices $\Lambda \in \R^{2n}$ for which $\cG(\Lambda,\phi)$ is a frame
is widely open. Some results for higher dimensional Gabor systems were obtained in
\cite{Groch2}, \cite{GroLyu}. 

\smallskip

In two-dimensions with Gaussian window function, an especially useful aspect of Gabor 
filters is their ``sensitivity to direction", which makes them especially useful in boundary
detection in image analysis.  This property, together with the observation that Gabor
filters well match the shape of receptor profiles in the $V_1$ visual cortex lead to
mathematical models of the connectivity and the signal analysis of the visual cortex in
terms of Gabor frames, \cite{Pet}, \cite{SaCiPe}, \cite{LiMa}. Another interesting
aspect emerges from the use of Gabor analysis in models of the visual cortex,
namely the question of signal analysis through Gabor filters applied to signals that
naturally live on curved smooth manifolds rather than on flat linear spaces. In the case
of the curved retinal surface, this leads to an interesting interplay between the contact
geometry of the bundle of contact elements of a curved surface, which models the
hypercolumn geometry of the $V_1$ visual cortex, and an associated family of
Gabor frames that carry out signal analysis on a consistent choice of local linearizations. 

\smallskip

In this paper, we plan to generalize the geometric setup developed for
the study of the visual cortex and address a broader type of question.

\begin{prob}\label{Problem}{\rm  
Given a signal on a curved smooth manifold of dimension $n$, 
provide a construction of Gabor frames that encodes
consistent local linearizations
(local mapping to vector spaces) with Gabor filters that 
are adapted to the detection of higher-dimensional boundaries. }
\end{prob}

By higher dimensional boundaries we mean here $(n-1)$-dimensional
hypersurfaces inside the $n$-dimensional manifold where the signal
has either a jump discontinuity or undergoes very rapid change (a 
smooth approximation to a jump discontinuity). 

\smallskip

As a motivation for this problem, consider the following type of application,
that will be described more in detail in \S \ref{ApplSec} below. 

\smallskip

Consider a mechanism, such as a robot $R$, whose possible movements in the ambient $3$-dimensional
space are parameterized by a configuration space $\cM(R)$, which is usually describable
as a manifold of some higher dimension $N=\dim \cM(R)$. For a general account of geometric
and topological robotics see for instance \cite{Farb}, \cite{Selig}. The configuration space 
$\cM(R)$ describes the possibilities and the constraints on motion that are intrinsic to the
mechanism itself. In addition to that, one may need to consider further constraints that
come from the interaction with the environment, and which manifest themselves as 
exclusion-zones that cut off certain regions of the configuration space that are not achievable
under the given environmental constraints. Moreover, one can also consider soft-constraints 
given by probability distributions over the manifold $\cM(R)$, which instead of realizing rigid constraints
excluding parts of the configuration space, give a degree of preference for certain motions and
configurations over others, for example for the purpose of motion planning. 

\smallskip

If one considers probability distributions over the manifold $\cM(R)$ that are absolutely continuous
with respect to the volume form given by an assigned Riemannian metric on $\cM(R)$, one can
identify such measures with their density function $\mu: \cM(R)\to \R_+$. This can be thought
of as a higher dimensional ``signal" that lives on the non-flat manifold $\cM(R)$. The case of
a sharp constraint excluding parts of the configuration space corresponds to the case where
$\mu$ is a (normalized) characteristic function of a subregion. 

\smallskip

One can see how in such a setting it is useful to be able to efficiently encode and transmit
the datum $\mu: \cM(R)\to \R_+$ using appropriate signal-analysis filters, and how an
especially important part of the information contained in the density function $\mu$ is
the location of sharp codimension-one boundaries (walls delimiting the accessible
region in configuration space). 

\smallskip

This problem exhibits all the characteristics described in more general terms above:
the signal $\mu$ lives on a higher-dimensional space $\cM(R)$ that is not a flat linear
space but a curved manifold. The nontrivial topological and geometric structure of $\cM(R)$
is important, as it describes the a priori possibilities of motion, so it needs to be taken
into account in the signal analysis process. The filters for signal analysis, on the other
hand, necessarily live on a flat linear space (or more generally a locally compact 
abelian group), where the convolution products happen that generate
the coefficients encoding the signal. The configuration space $\cM(R)$ in general
does not carry such group or vector space structure. Thus, over the curved manifold $\cM(R)$ one needs
a consistent system of linearizations and filters that can be used to encode $\mu$, in a
way that also accuonts for the underlying geometry of $\cM(R)$. Moreover, one needs the
choice of filters to be suitable for detecting the presence of $(N-1)$-dimensional hypersurfaces
in $\cM(R)$ that constitute sharp boundaries in the signal $\mu$. In the following section we
describe a general procedure for approaching this problem.

\section{Geometry of higher-dimensional signal analysis}\label{HigherSignal}

In this section we construct a system of Gabor filters associated to a smooth Riemannian $n$-dimensional
manifold $B$, using the associated manifold $M$ of contact elements as the
parameterizing space for the Gabor filters, where the fibers of the bundle projection
from the manifold of contact elements to the base manifold $B$ encode
the possible choices of $(n-1)$-dimensional hyperplanes in the (co)tangent
spaces of $B$, which provide the boundary sensitivity property of the Gabor
frames. Signals on $B$ are lifted to signals on a 
bundle $\cE$ over $M$ given by the pullback of the tangent bundle $TB$ along
the projection $M\to B$. Passing from a signal on $B$ to a signal on $TB$ amounts 
to replacing $B$ by its local linearizations and storing the signal as a consistent collection of
lifts to these linearizations. Further passing from $TB$ to its pullback $\cE$ on $M$ maintains
the same fibers where the signal is stored, but allows for sections that have
a further dependence on the variables in the fiber $\bS^{n-1}$ of the
fibration $\bS^{n-1}\hookrightarrow M \to B$, that is, for filters with a boundary-detection
capacity. The window function for the filter construction incorporate this dependence 
on hyperplanes in the form of a twist of a multivariate Gaussian by a phase factor that depends
on a point in the fibers $\bS^{n-1}$, that is, on an oriented hyperplane.
The lattice that determines the Gabor
filters from the window function is obtained in the form of a bundle of lattices
in the bundle $\cE\oplus \cE^\vee$, which is constructed using the contact
geometry of the $(2n-1)$-dimensional manifold $M$ of contact elements, and
the $SO(n)$-action on the cosphere bundle of $B$.

\subsection{The manifold of contact elements}

We first discuss some general facts of Riemannian geometry that will be
useful for our main construction.

\smallskip

Given any smooth manifold $B$ of arbitrary dimension $n$, the cotangent bundle $T^*B$
is equipped with a canonical Liouville $1$-form $\lambda$,
which is intrinsically defined by  
$$\lambda_{(b,p)}(v)=p_b(d\pi(v))$$
for vectors $v\in T_{(b,p)}(T^*B)$, with $\pi: T^*B\rightarrow B$ the projection map.

\smallskip

Suppose $B$ is also equipped with a Riemannian metric $g$.  
Let $\mathbb{S}(T^*B)$ denote the cosphere bundle, that is, the 
unit sphere bundle of the cotangent bundle $T^*B$.
We write $\pi_\bS:\mathbb{S}(T^*B)\rightarrow B$ for the 
induced projection map. When no confusion arises, we will just use $\pi$ for both
the projection $\pi_\bS$ on $\mathbb{S}(T^*B)$ and the 
projection $T^*B\to B$. Each
fiber $\pi^{-1}(b)\subset \mathbb{S}(T^*B)$, for $b\in B$, is isomorphic to the unit 
sphere $\mathbb{S}^{n-1}$ 
through a (non canonical) isomorphism $j_b:\pi^{-1}(b)\rightarrow \mathbb{S}^{n-1}$. 

\smallskip

The 1-form $\alpha$ on $\mathbb{S}(T^*B)$, induced by the Liouville form $\lambda$, is a contact $1$-form. 

\begin{defn}\label{Mcontactelts}
The $(2n-1)$-dimensional manifold $M=\mathbb{S}(T^*B)$, endowed with the contact form $\alpha$,
is the {\rm manifold of contact elements} of the $n$-dimensional $B$.
\end{defn}

The manifold $M$ can be understood as describing, for each point in the base manifold $B$,
all the possible choices of a ``hypersurface direction" passing through that point
(that is, of an oriented tangent hyperplane). Indeed the collection of all hyperplanes
in $\R^n$ is parameterized by the real projective space $\bP^{n-1}(\R)=\R^n\smallsetminus \{0\}/\R^*$,
while the oriented hyperplanes in $\R^n$ are parameterized by the covering of $\bP^{n-1}(\R)$ given
by the sphere $\bS^{n-1}$. Since $B$ is a Riemannian manifold, the Riemannian metric $g_B$ provides
a (non-canonical) isomorphism between tangent and cotangent bundles 
$$ g_B : TB \stackrel{\simeq}{\to} T^* B\, , \ \ \ \ \  g_B: v \mapsto g_B(v,\cdot) \, . $$
Thus, we can interpret the fiber $\bS^{n-1}$ of $M=\bS(T^*B)$ at a point $b\in B$ 
as parameterizing  hyperplanes in either $T^* _b B$ or $T_bB$.

\subsubsection{Action by rotations}
The group $SO(n)$ of orientation preserving orthogonal transformations of $(\mathbb{R}^n, \langle - , - \rangle )$  acts transitively on $\mathbb{S}^{n-1}\subset \mathbb{R}^n$,  
$$ SO(n)\times \mathbb{S}^{n-1}\rightarrow \mathbb{S}^{n-1} \, ,  \ \ \ \ \ 
(A,p_1,...,p_n)\mapsto A\cdot (p_1,...,p_n) \, . $$
Suppose $B$ is oriented, then $F_{SO(n)}(T^*B)$ denotes the principal $SO(n)$-bundle of positively oriented orthonormal frames of $T^*B$ with respect to the bundle metric
\begin{equation*}
    g^*_b(p_1,p_1)=g_b(g_B^{-1}p_1,g_b^{-1}p_2)\, , \ \ \text{ for }p_1,p_2 \text{ in }T^*_bB \, .
\end{equation*}

The unit cotangent bundle $\mathbb{S}(T^*B)$ is the bundle associated to the action of $SO(n)$ on $\mathbb{S}^{n-1}$ 
\begin{equation*}
   F_{SO(n)}(T^*B)\times_{SO(n)} \mathbb{S}^{n-1} / \, SO(n) =\mathbb{S}(T^*B).
\end{equation*}
The bundle $F_{SO(n)}(T^*B)$ admits a left action of the group $SO(n)$, 
$$SO(n)\rightarrow {\rm Aut}(F_{SO(n)}(T^*B))\, ,\ \ \ \ \ A\mapsto \nu_A$$
Then, each map $\nu_A\times id:F_{SO(n)}(T^*B)\times \mathbb{S}^{n-1}\rightarrow F_{SO(n)}(T^*B)\times\mathbb{S}^{n-1}$, with $A\in  SO(n)$, induces a map on the quotients,
$$F_A:\mathbb{S}(T^*B)\rightarrow\mathbb{S}(T^*B) \, , \ \ \ \ \ 
(b,p)\mapsto (b,A^{-1}(p))\, , $$
therefore one can consider the corresponding action of $SO(n)$ on the sphere bundle $\mathbb{S}(T^*B)$, 
given by 
$$SO(n)\times \mathbb{S}(T^*B)\rightarrow \mathbb{S}(T^*B)\, , \ \ \ \ \ \ 
(A,(b,p))\mapsto F_A(b,p)=(b, A^{-1}(p))\, . $$

\begin{lem}
The map $$SO(n)\rightarrow {\rm End}(\mathbb{S}(T^*B))\, , \ \ \ \ \ 
A\mapsto F_A$$ is a smooth right group action of $SO(n)$ on $\mathbb{S}(T^*B)$. 
Additionally, the orbits of $SO(n)$ are exactly the fibers of $\pi_{\mathbb{S}}:\mathbb{S}(T^*B)\rightarrow B$.
\end{lem}

\begin{proof}
 Firstly, $SO(n)$ acts smoothly on $F_{SO(n)}(T^*B)\times \mathbb{S}^{n-1}$ by $A\mapsto \nu_A\times id$ and the quotient map $$q:F_{SO(n)}(T^*B)\times \mathbb{S}^{n-1}\rightarrow \mathbb{S}^{n-1}(T^*B)\, ,\ (f,p)\mapsto [(f,p)]$$ is smooth. Therefore, the action map $$SO(n)\times\mathbb{S}(T^*B)\rightarrow \mathbb{S}(T^*B)\, ,\  (A, [(f,p)])\mapsto F_A([f,p])=q\circ (\nu_A\times id)(f,p)$$ 
 is smooth.
 
 To prove that the map $A\mapsto F_A$ is an antihomomorphism we consider the following: 
 for $A_1, A_2\in SO(n)$ and $(b,p)\in \mathbb{S}(T^*B)$,
$$F_{A_1A_2}(b,p)=(b,(A_1A_2)^{-1}(p))=(b,A_2^{-1}A_1^{-1}(p))=F_{A_2}F_{A_1}(b,p)\, .$$ 

 Finally, since the action of $SO(n)$ on $F_{SO(n)}\times \bS^{n-1}$ is transitive and $q$ is surjective, the orbits of $SO(n)$ on $\bS(T^*B)$ are exactly the fibers $\pi^{-1}(b)$ for all $b\in B$. 
\end{proof}

\subsection{Multiple contact forms}
The above action of $SO(n)$ induces an action on the sections of the cotangent bundle of $\mathbb{S}(T^*B)$ 
$$SO(n)\times \Gamma( T^*(\mathbb{S}(T^*B)))\rightarrow \Gamma(T^*(\mathbb{S}(T^*B)))$$
$$(A,\alpha)\mapsto F_A^* \alpha \, ,$$
where by $\Gamma( T^*(\mathbb{S}(T^*B)))$ we denote the space of smooth co-vector fields. The maps $F_A$ have the property  $\pi \circ F_A= \pi$. Abusing the notation slightly, we are going to consider $F_A$ to be $F_A=id\times A$ for the proof of the following lemma.  

\begin{lem}\label{FAalpha}
For any  map $F_A:\mathbb{S}(T^*B)\rightarrow \mathbb{S}(T^*B)$, with $A\in SO(n)$, 
the standard contact form $\alpha$ on $\mathbb{S}(T^*B)$ has the property that 
$$(F_A^* \alpha)_{(b,p)}(X_{(b,p)})=\alpha_{(b,Ap)}(X_{(b,Ap)}), ~X\in \mathfrak{X}(\bS(T^*B))$$  and $F_A^* \alpha$ is a contact $1$-form on 
$\mathbb{S}(T^*B)$.
\end{lem}

\begin{proof}
 The contact 1-form $\alpha$ on $\mathbb{S}(T^*B)$ is induced by the Liouville form $\lambda$ on $T^*B$ as 
 $\iota^*(\lambda_{(b,1,p)})=\alpha_{(b,p)}$, through the inclusion $$\iota: \bS(T^*B)\rightarrow T^*B,  
 ~(b_1,...,b_n,p_2,...,p_n)\mapsto (b_1,...,b_n,1,p_2,...,p_n).$$ Then, if $A$ is any $A\in SO(n)$ and 
  $F_A: \mathbb{S}(T^*B)\rightarrow  \mathbb{S}(T^*B)$ its induced map, then  
 for $X\in T_{(b,p)}\mathbb{S}(T^*B)$ we have
\begin{equation}\label{pullback}
(F_A^* \alpha)_{(b,p)}(X_{(b,p)})=\alpha_{F(b,p)}(dF_A X_{\mid_{(b,p)}})=\iota^*(\lambda_{(b,A(\iota(p)))})(dF_A X_{\mid_{(b,p)}})\end{equation}
Additionally, the diagram 
\begin{center}
\begin{tikzcd}
\mathbb{S}(T^*B) \arrow[r, "\iota"] \arrow[d, "\pi_{\bS}"] 
& T^*B \arrow[d, "\pi"] \\ 
B \arrow[r, "="]
&  B
\end{tikzcd} 
\end{center}
commutes and therefore $\iota^*\circ\pi^*=\pi_\bS^*$ and \eqref{pullback} becomes
$$\iota^*(\pi^*(A(\iota(p))_b))(dF_A X_{\mid_{(b,p)}})=\pi_\bS^*(A(\iota(p)_b)( dF_A X_{\mid_{(b,p)}})
=(A(\iota(p))_b)(d\pi_\bS dF_A X_{\mid_{(b,p)}}) $$ $$ =(A(\iota(p))_b)(d\pi_\bS X_{\mid_{(b,Ap)}})
=\pi_\bS^*(A(\iota(p))_b)(X_{(b,Ap)})
=\iota^*\pi^*(A(\iota(p))_b))(X_{(b,Ap)}) $$ 
$$ =\iota^*(\lambda_{(b,A(i(p)))})(X_{(b,Ap)})
=\alpha_{F_A(b,p)}(X_{(b,Ap)})\, .$$
Moreover, since $F_A$ is a diffeomorphism for every $A \in SO(n)$ it follows that $F^*_A a$ is a contact 1-form. 
\end{proof}

 \smallskip

\begin{prop}\label{Fialpha}
On the co-sphere bundle $M=\mathbb{S}(T^*B)$ of an $n$-dimensional manifold $B$ 
there is a collection $\{ F^*_i \alpha  \}_{i=1}^{n-1}$  of $1$-forms 
which satisfy the conditions:
\begin{enumerate}
    \item All the $F^*_i\alpha$, $i=1,\ldots, n-1$, are contact $1$-forms on $M$, with $F^*_0\alpha=\alpha$.
    \item At each point $(b,p)\in M$, the co-vectors 
    $$\alpha_{(b,p)} ,F^*_1\alpha_{(b,p)},F_2^*\alpha_{(b,p)}, \ldots ,F^*_{n-1}\alpha_{(b,p)}$$ 
    are orthogonal with respect to the inverse Riemannian metric tensor on $B$ (the Riemannian
    metric on $T^*_{(b,p)}B$).
    \item The fibers $\{\pi^{-1}(b), b\in B\}$ of the fiber bundle $\pi:\mathbb{S}(T^*B)\rightarrow B$ are 
    Legendrian submanifolds in each of the contact distributions $\xi_i$ induced by the contact $1$-form $F^*_i\alpha$, for each $i=0,\ldots, n-1$.
    \item\label{linind} If $V\rightarrow M$ is the vertical tangent bundle of $
    \pi:M\rightarrow B$, $\Gamma(V)$ denotes the vertical vector fields of $M$ and $\mathfrak{R}_{F^*_i a}$ is the Reeb field of the contact 1-form $F^*_ia$ for $i=0,...,n$, then $span\{\mathfrak{R}_{F^*_ia}\}\cap \Gamma(V)=
    \{0\}$. Moreover,  for every
$m\in M$ the vectors $\{\mathfrak{R}_{F^*_i a}(m): i=0,...,n-1\}$ are linearly independent.  
     
\end{enumerate} 
  \end{prop}
  
\begin{proof} Let $(b,p)$ a point in $M$ such that $p$ is in the fiber $
\pi^{-1}(b)$ over $b\in B$ and $(U,
\Phi)$ a coordinate chart around $(b,p)$. Since the action of $SO(n)$ on the fiber $\pi^{-1}(b)$ is transitive, there exist $R_1,...,R_{n-1}\in SO(n)$ such that $\{
\Phi(p),\Phi(F_{R_1}(p)),...,\Phi(F_{R_{n-1}}(p))\}$ is an orthogonal basis of $\bR^n$  with respect to the metric induced by the Riemannian metric on $T^*_{(b,p)}B$ and $\Phi^{-1}:\bR^n\rightarrow U$. The $1$-forms $F_i^*\alpha=F_{R_i}^*\alpha$, for $i=0,\ldots ,n-1$,  that are induced by these rotations are mutually orthogonal contact forms, since by Lemma~\ref{FAalpha} we have that  
\begin{equation}\label{Fialphas}
F_i^* \alpha_{(b,p)} :=\alpha_{(b,R_i(p))}=\pi^*_{(b,p)}(R_i(p)).
\end{equation}
Now, we need to show that the fibers $\{\pi^{-1}(b),~b\in B\}$ are Legendrian submanifolds in all the contact distributions $\xi_i$ associated to the contact forms $F_i^* \alpha$. Consider a vector
$v_p\in T_p \pi^{-1}(b)$. Written in canonical local coordinates $(U,\Phi(b,p)=(b_1,...,b_n, q_1,...,q_{n-1}))$, $v_p$ takes the form
$$ v_p=f_1(b,p)\partial_{q_1}+...+f_{n-1}(b,p)\partial_{q_{n-1}}.$$
Thus, $$F_i^*\alpha_{(b,p)}(v)=\alpha_{(b,R_ip)}(dR_i v)=\pi^*(R_i p)(dR_i v)=R_i p(d\pi (dR_i v))
=R_i p(0)=0\, . $$
for all $i=0,...,n-1$ and therefore $T_p \pi^{-1}(b)\subset {\rm Ker}(F_i^*\alpha)$ for all $i=0,\ldots,n-1$.

Finally, it is left to prove part \ref{linind}. On the coordinate chart $(U, (b_1,...,b_n, q_1,...,q_{n-1}))$ 
any contact $1$-form $F_i^* \alpha$ can be written in local coordinates as
$$F_i^*\alpha= a_{i,1}(q_1,...,q_{n-1})\, db_1+ \cdots + a_{i,n}(q_1,...,q_{n-1})\, db_n\, ,$$
since $$\alpha=q_1\,db_1+...+\sqrt{1-q_1^2- \cdots -q_{n-1}^2}\,db_n \ \ \ \text{ and } \ \ \ $$ $$ (a_{i,1},\ldots a_{i,n})=R_i(q_1,\ldots, q_{n-1}, \sqrt{1-q_2^2- \cdots -q_{n-1}^2}). $$ Thus, from the defining property 
of the Reeb field,  it follows that 
$$ \iota_{\fR_{F^*_i \alpha}} dF^*_i \alpha(X)=0\ \ \text{for all vector fields } X \text{ in } TM\, .$$
Writing the Reeb fields $\mathfrak{R}_{F^*_i a}$ in coordinates $\fR_{F^*_i \alpha}=R^{b_1}_i\, \partial b_1+ \cdots +R^{b_{n-1}}_i\, \partial b_{n-1}+R^{q_1}_i \,\partial q_1+ \cdots +R^{q_n}_i\, \partial q_n$
it follows that 
$$\begin{pmatrix}\frac{\partial a_{i,1}}{\partial q_1}&...& \frac{\partial a_{i,1}}{\partial q_{n-1}}\\
\vdots &\ddots &\vdots\\
\frac{\partial a_{i,n}}{\partial q_1}&...&\frac{\partial a_{i,n}}{\partial q_{n-1}}
\end{pmatrix}\begin{pmatrix}R_{a_i}^{q_1}\\
\vdots\\
R_{a_i}^{q_{n-1}}\end{pmatrix}=0,$$
 and therefore 
$$ \fR_{F^*_i \alpha} =R^{b_1}_i \, \partial b_1+ \cdots +R^{b_n}_i\, \partial b_n. $$
\end{proof}

\subsection{The bundle of signal spaces}

A signal defined over $\mathbb{R}^n$ is a function $f:\mathbb{R}^n\rightarrow \mathbb{R}$ which belongs to some function space, for example the Hilbert space $L^2(\R^n)$ with the Lebesgue measure. 
Signal analysis is performed through a family of filters $\{ \phi_\alpha \}_{\alpha\in A}$ which are integrated against the signal to obtain corresponding coefficients
$$ c_{\alpha}=\int \phi_\alpha f \, . $$
When the family of filters satisfies certain conditions, such as the frame condition
of Gabor filters or the case of orthonormal bases, the signal can be fully reconstructed 
from this set of coefficients $c_\alpha$.

\smallskip

The setting for signal analysis through Gabor filters requires the underlying
linear space structure, as the Gabor system of filters is constructed via the
linear operations of translation and modulation. When the signal itself
is defined on an $n$-dimensional manifold $B$ rather than on a vector space $\R^n$,
applying Gabor signal analysis requires the use of local linerizations of the
manifold $B$, over which one can construct Gabor filters. These linearizations
are provided by the tangent spaces (translation coordinates) and their duals
(modulation coordinates). A special case of this kind of Gabor signal analysis on
manifolds, motivated by models of the visual cortex in neuroscience, 
was introduced in \cite{LiMa}, in the case where $B$ is a 
$2$-dimensional Riemann surface. We generalize here the construction to
arbitrary dimensions. The key observation of \cite{LiMa} is that it is important
to maintain the distinction between the  coordinates of the curved manifolds $B$
and $M=\bS(T^*B)$ and the linear coordinates in the fibers of the bundles
$TB$ and $T^*B$. Making this distinction precise geometrically requires introducing
a suitable {\em bundle of signal spaces}. This generalizes the bundle of
signal planes introduced in \cite{LiMa}.

\smallskip

\begin{defn}\label{defEEvee}
The bundle of signal spaces $\mathcal{E}$ is the real $n$-space bundle on the contact $(2n-1)$-dimensional
manifold $M=\mathbb{S}(T^*B)$ obtained by pulling back the tangent bundle $TB$ 
of the $n$-dimensional manifold $B$ to $M$ along the projection 
$\pi_\bS :\mathbb{S}(T^*B)\rightarrow B$ of the unit sphere bundle of $T^*B$,
\begin{equation}\label{Ebundle}
    \mathcal{E}: =\pi_\bS^* \, TB \, .
\end{equation}
Let $\mathcal{E}^\vee$ denote the dual bundle of $\mathcal{E}$
\begin{equation}
    \mathcal{E}^\vee=\Hom(\mathcal{E}, \mathbb{R})
\end{equation}
\end{defn}

\smallskip

The real vector bundle $\cE$ of rank $n$ and its dual determine a rank $2n$ vector bundle
over the $2n-1$ dimensional manifold $M$, given by their direct sum
\begin{equation}\label{EoplusEvee}
\cE \oplus \cE^\vee \, . 
\end{equation}
The spatial frequencies (modulation operators) are represented by the fiber coordinates of the 
dual bundle $\mathcal{E}^\vee$, while the translation operators are provided by
the fiber coordinates of the bundle $\cE$, see \S \ref{WindowSec} below. 

\smallskip

The $L^2$ space $L^2(\mathcal{E},\mathbb{R})$ is 
determined by the condition
\begin{equation}\label{signalE}
\Big(\int_M \int_{\mathcal{E}_{(b,p)}} \mathcal{I}^2(u,b,p) \, dvol_{\mathcal{E}_{(b,p)}}(u) \, 
dvol_M(b,p)\Big)^{1/2}< \infty, 
\end{equation}
where $(b,p)$ are the local coordinates of $M$, and $u=(u_1,...,u_n)$ are the coordinates of the fibers 
$\mathcal{E}_{(b,p)}$. The norm on the fibers $dvol_{\mathcal{E}_{(b,p)}}$ is induced by the inner product on $TB$ through the pullback map and $dvol_M$ is the measure induced by the Riemannian volume form on $M$
determined by the Riemannian metric on $B$ and the round metric on the fibers $\bS^{n-1}$
of $M=\bS(T^* B)$.

\smallskip

\begin{defn}\label{defsignal}
A signal is a function $$\mathcal{I}:\mathcal{E}\rightarrow \mathbb{R},$$
on the bundle of signal spaces, 
with $\mathcal{I}\in L^2(\mathcal{E},\mathbb{R})$.
\end{defn}

\smallskip

Given a Riemannian manifold $B$, the exponential map is a {\em locally defined} map
$\exp : TB \to B$ from the tangent bundle of $B$ to the manifold $B$, where fiberwise
$\exp_b: T_b B \to B$ is obtained by considering, for a vector $v\in T_b B$ the
unique geodesic $\gamma_v$ in a neighborhood of $b$ in $B$ starting at $b$ with tangent vector $v$
and setting $\exp_b(v)=\gamma_v(1)$. The domain of definition of $\exp_b$ in $T_b B$
is a sufficiently small ball $B(0,R)$ around $0\in T_b B$ such that for all $v\in B(0,R)$ the
point $\gamma_v(1)\in B$ is uniquely determined by the existence and uniqueness theorem
applied to solutions of the geodesic equation for the Riemannian metric in $B$. The exponential
map is defined on all of $T_b B$, for all $b\in B$, iff the manifold $B$ is geodesically complete. 

\smallskip

At a given point $b\in B$ let $R_{inj}(b) >0$ denote the supremum of all the radii $R>0$ such that the  
exponential map $\exp_b$ is a diffeomorphism on the ball $B(0,R)$ of radius $R$ in 
$T_b B$ to its image in $B$. For a compact manifold $B$, this determines a continuous 
{\em injectivity radius function} $R_{inj}: B \to \R_+^*$. We denote by $\cB(T_b B)$
the ball $B(0,R_{inj}(b))$ in the tangent space $T_b B$. Under the pullback from $B$
to $M$ we obtain a collection of balls of radius $R_{inj}(b)$ in each fiber $\cE_{(b,p)}$
with $(b,p)\in M$. We denote these balls by $\cB(\cE_{(b,p)})$. 

\smallskip

\begin{lem}\label{MsignalEsignal}
Let $B$ be a compact smooth manifold.
Let $f: B \to \R$ be a smooth function (or more generally a function in
$L^\infty(B)$ with the measure given by the volume form
of the Riemannian metric). Then $f$ determines a signal $\cI(f)\in L^2(\mathcal{E},\mathbb{R})$,
with the property that $f$ can be recovered from the restrictions $\cI(f) |_{\cB(\cE_{(b,p)})}$.
\end{lem}

\proof Consider then a smooth function $\chi: \cE \to \R$, such that the
restriction $\chi_{(b,p)}:=\chi|_{\cE_{(b,p)}}$ to the fiber $\cE_{(b,p)}$
is a rapidly decaying Schwartz function $\chi_{(b,p)}: \cE_{(b,p)} \to \R$,
which satisfies $\chi_{(b,p)}\equiv 1$ inside the ball $\cB(\cE_{(b,p)})$. 

Compact manifolds are geodesically complete, hence the exponential
map of $B$ is defined on the full tangent spaces, not just on a neighborhood
of the origin, hence the function $f: B \to \R$ determines by precomposition
a function $f\circ \exp: TB \to R$. Since $f$ is bounded (respectively,
essentially bounded) on $B$, the pullback is bounded (respectively,
essentially bounded) on $TB$.  

Consider then the pullback diagram
$$ \xymatrix{ \cE \ar[r]^h \ar[d]^{\pi_\cE} & TB \ar[d]^{\pi_{TB}} \\
M \ar[r]^{\pi_\cS} & B } $$
with $\pi_\bS: M \to B$, $\pi_{TB}: TB \to B$, and $\pi_\cE: \cE \to M$ the projections,
and define the function $\cI(f): \cE \to \R$ as 
$$ \cI(f) := \chi \cdot f \circ \exp \circ h \, . $$
The function $f \circ \exp \circ h$ is bounded (or essentially bounded) on the fibers 
$\cE_{(b,p)}$ and $\chi$ is rapidly decreasing, hence their product is in $L^2(\cE_{(b,p)})$.
This fiberwise $L^2$ norm varies smoothly with the point $(b,p)\in M$, and is bounded
on the compact manifold $M$, so that the integration along $M$ in \eqref{signalE} 
is also finite. Thus $\cI(f)\in L^2(\mathcal{E},\mathbb{R})$. 

Moreover, since we have $\chi_{(b,p)}\equiv 1$ inside the ball $\cB(\cE_{(b,p)})$,
the restrictions satisfy $\cI(f) |_{\cB(\cE_{(b,p)})}=f \circ \exp \circ h$. Since the
map $h: \cE\to TB$ is the fiberwise identification $\cE_{(b,p)}\simeq T_b B$, for
all $(b,p)\in \pi_\bS^{-1}(b)$, and the map $\exp\circ h$ is a diffeomorphism
when restricted to $\cB(\cE_{(b,p)})$ since $\exp$ is a diffeomorphism on 
$B(0,R_{inj}(b)) \subset T_b B$. Thus, the restriction of $f$ to $\exp(B(0,R_{inj}(b)))$ 
can be fully reconstructed from $\cI(f) |_{\cB(\cE_{(b,p)})}$ and the collection of
the functions $\cI(f) |_{\cB(\cE_{(b,p)})}$ fully determine $f$. 
\endproof

\smallskip

Lemma~\ref{MsignalEsignal} shows that it is equivalent to think of a signal
as a function $f: B \to \R$ defined over the $n$-dimensional manifold $B$,
or as the corresponding $\cI(f): \cE \to \R$, that is, as a function on the
bundle of signal spaces as in Definition~\ref{defsignal}. Passing from
$f: B \to \R$ to $\cI(f): \cE \to \R$ corresponds to replacing a signal on a
manifold with a consistent collection of signals on its local linearizations.

\begin{figure*}[btp]
\centering
\includegraphics[scale=0.4]{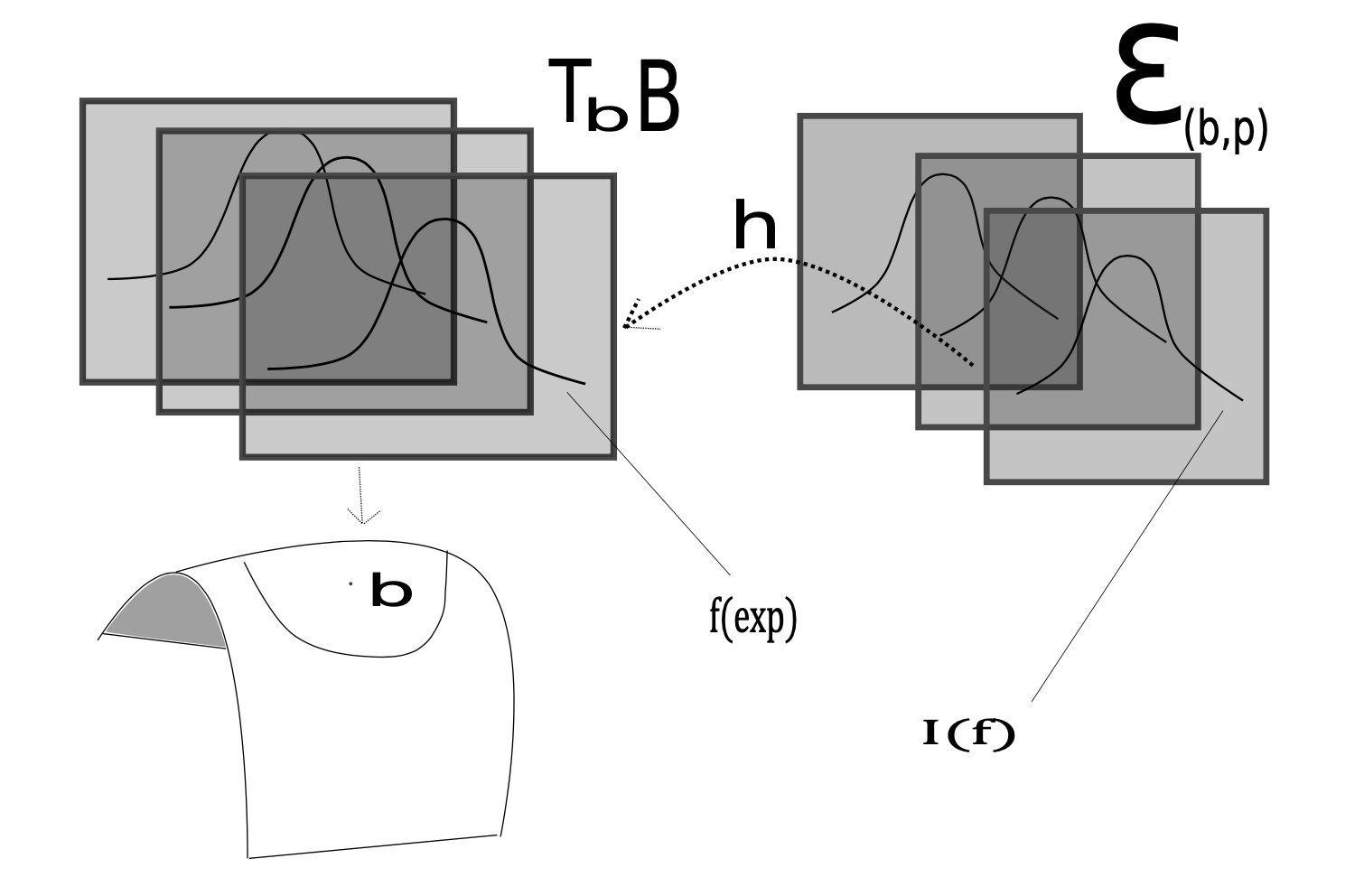}
\caption{\small Passing from
$f: B \to \R$ to $\cI(f): \cE \to \R$ corresponds to replacing a signal on a
manifold with a consistent collection of signals on its local linearizations.
\label{LiftedSignal}}
\end{figure*}

\newpage

\subsection{The lattice bundle}\label{LattSec}

We use the data of the signal bundle $\cE$ and its dual $\cE^\vee$, along with
the contact $1$-forms $F_i^* \alpha$ and their Reeb vector fields $\fR_{F_i^* \alpha}$
discussed above, to obtain an additional structure on the bundle $\cE\oplus \cE^\vee$
consisting of a ``lattice bundle", in the sense of the following definition.

\smallskip

\begin{defn}\label{latticebdle}
Given a vector bundle $\cF$ over a smooth manifold $M$, a framed lattice bundle
structure on $\cF$ is a subbundle $\Lambda \subset \cF$ with discrete fibers 
$\Lambda_m=\pi_\Lambda^{-1}(m)$, for $m\in M$, where $\pi_\Lambda
=\pi_\cF |_\Lambda$ is the restriction of the projection $\pi_\cF: \cF \to M$, with
the property that there is a dense open set $\cU \subset M$ such that, for all
$m\in \cU$ the fiber $\Lambda_m\subset \cF_m$ is a lattice in the vector space
$\cF_m=\pi_\cF^{-1}(m)$, endowed with a generating set (frame). 
\end{defn}

\smallskip

\begin{lem}\label{lemLambda}
The contact $1$-forms $\{ F_i^* \alpha \}_{i=0}^{n-1}$ on the manifold $M=\bS(T^*B)$ of an
$n$-dimensional smooth manifold $B$ determine a framed lattice bundle
structure on the rank $2n$ vector bundle $\cE \oplus \cE^\vee$ of \eqref{EoplusEvee}. 
\end{lem}

\proof Consider the set of $1$-forms $\{ F_i^* \alpha \}_{i=0}^{n-1}$ as sections of $T^* B$,
and the set of associated Reeb vector fields $\fR_{F_i^* \alpha}$, as sections of $TB$.
Under the identification of the fibers 
$\cE_{(b,p)}\simeq T_b B$, for all $(b,p)\in \pi_\bS^{-1}(b)$, induced by the pullback map
$h: \cE \to TB$, we can identify the $\fR_{F_i^* \alpha}$ as sections of $\cE$ by
precomposition with the projection $\pi_\bS: M \to B$ and postcomposition with
the identification $T_b B \stackrel{\simeq}{\to} \cE_{(b,p)}$. Similarly, we can
identify the $F_i^* \alpha$ with sections of $\cE^\vee$. For simplicity, we maintain the same
notation $\fR_{F_i^* \alpha}$ and $F_i^* \alpha$ for these resulting sections without
writing the pullback maps explicitly. We can then consider the spans
\begin{equation}\label{LambdaEdef}
\begin{array}{l}
 \Lambda_\cE:= {\rm span}_\Z \{ \fR_{F_i^* \alpha} \} \subset \cE \, , \\[3mm]
 \Lambda_{\cE^\vee} := {\rm span}_\Z \{ F_i^* \alpha \} \subset \cE^\vee \, , \\[3mm]
 \Lambda:= \Lambda_\cE \oplus \Lambda_{\cE^\vee} \subset \cE \oplus \cE^\vee \, . 
\end{array}
\end{equation} 
The rank of the fibers $\Lambda_m$ is $2n$ on a dense open set $\cU=\pi_\bS^{-1}(\cV)$
where $\cV\subset B$ is the dense open set where none of the $1$-forms $F_i^* \alpha$ and
of the vector fields $\fR_{F_i^* \alpha}$ vanish. Thus, it is a framed lattice bundle
structure on $\cE \oplus \cE^\vee$ with frame given by the set 
$\{ \fR_{F_i^* \alpha}, F_i^* \alpha \}_{i=0}^{n-1}$.
\endproof

\smallskip

\begin{rem}\label{nodual}{\rm
Unlike the case where $B$ is a $2$-dimensional Riemann surface, considered 
in \cite{LiMa}, in this more general case the Reeb fields $\fR_{F_i^* \alpha}$
are in general not in the kernel of the forms $F_j^* \alpha$ with $j\neq i$, that is,
one does not have the ``dual-basis relation" 
$$  \langle F_j^* \alpha, \fR_{F_i^* \alpha} \rangle =\delta_{ij}. $$
}\end{rem}

\smallskip

It is possible to consider another framed lattice bundle
structure on $\cE \oplus \cE^\vee$, where the frame 
$\{ \fR_{F_i^* \alpha}, F_i^* \alpha \}_{i=0}^{n-1}$ of Lemma~\ref{lemLambda}
is replaced by another frame $\{ \fV_i , F_i^* \alpha \}_{i=0}^{n-1}$, 
where the sections $\fV_i$ of $\cE$ are taken to be the dual basis of the 
$F_i^* \alpha$ (seen as sections of $\cE^\vee$), so that the relation
$$  \langle F_j^* \alpha, \fV_i \rangle =\delta_{ij} $$
holds on the dense open set $\cU=\pi_\bS^{-1}(\cV)$, 
with $\cV\subset B$ the dense open set where all the $1$-forms $F_i^* \alpha$
are non-trivial.  In this case the lattice bundle structure on $\cE \oplus \cE^\vee$
is given by 
\begin{equation}\label{LambdaEdef2}
\begin{array}{l}
 \Lambda_\cE:= {\rm span}_\Z \{ \fV_i \} \subset \cE \, , \\[3mm]
 \Lambda_{\cE^\vee} := {\rm span}_\Z \{ F_i^* \alpha \} \subset \cE^\vee \, , \\[3mm]
 \Lambda:= \Lambda_\cE \oplus \Lambda_{\cE^\vee} \subset \cE \oplus \cE^\vee \, . 
\end{array}
\end{equation} 
with the frame $\{ \fV_i , F_i^* \alpha \}_{i=0}^{n-1}$.
We refer to the lattice bundle structure of \eqref{LambdaEdef} as the {\em Reeb
lattice bundle structure} and to \eqref{LambdaEdef2} as the {\em dual-basis 
lattice bundle structure}. 

\smallskip

We use the following notation for sections of lattice bundles.

\begin{defn}\label{notationLambda}
For a framed lattice bundle structure $\Lambda =\Lambda \oplus \Lambda^\vee$ on
a vector bundle $\cE\oplus \cE^\vee$ over $M$, we write $(\lambda,\lambda')$ for sections
$$ (\lambda,\lambda') \in \Gamma(M, \Lambda) \, , $$
that is, for sections of $\cE\oplus \cE^\vee$ with values in the discrete subbundle $\Lambda$.
\end{defn}

\subsection{The choice of window function}\label{WindowSec}

In the usual setting of time-frequency analysis of signals $\cI \in L^2(\R)$, the signal
analysis is performed through linear transform using a family of wavelets $\{\phi_\alpha\}_{\alpha\in A}$. 
These wavelets are commonly referred to as ``time-frequency atoms", because they play the role 
of building blocks for the signal. In the discrete case, the family of atoms is obtained through 
a discrete family $\{ M_w T_x \}_{\{(x,w)\}\subset \Lambda}$ of time-frequency shift operations 
(translations $T_x$ and modulations $M_w$)
$$ M_w T_x:  L^2(\mathbb{R}^d)\rightarrow L^2(\mathbb{R}^d)$$
$$\phi(t)\mapsto e^{-2\pi i t\cdot w}\phi(t-x) \, , $$
where $\Lambda\subset \R^{2n}$ is a discrete set (such as a lattice or a periodic set).
The wavelets are well concentrated in time and frequency, so they are often referred to as ``window functions".
A Gabor system $\cG(\phi,\Lambda)$, with a given choice of window function $\phi \in L^2(\R^n)$ and
lattice $\Lambda\subset \R^{2n}$, is the resulting collection of wavelets of the form
$$  \cG(\phi,\Lambda) := \{ M_w T_x \phi \,|\, (x,w)\in \Lambda \} \, . $$

\smallskip

In our setting, the construction of wavelets takes place on the $n$-dimensional linear spaces $\cE_{(b,p)}$
given by the fibers of the bundle $\cE$ of signal spaces introduced in Definition~\ref{defEEvee}, with
translation and modulation operators associated to the bundle of lattices introduced in \S \ref{LattSec}
above.  We describe here the appropriate choice of window function and the construction of the resulting
Gabor systems on the bundle $\cE$.

\smallskip

Let $V$ and $\eta$ denote the variables in the fibers $T_b B$ and $T^*_b B$ respectively,
for $b\in B$, with $\langle \eta,V \rangle$ denoting the duality pairing of $TB$ and $T^*B$. 

\begin{defn}\label{TBwindow}
A window function on the bundle $TB \oplus T^*B$ is a smooth real-valued 
function $\Phi : TB \oplus T^*B \to \R$ from the total space of the vector bundle, 
defined fiberwise as
\begin{equation}\label{PhiTB}
    \Phi_b (V,\eta)=\exp\big(-V^t A_b V-i \langle \eta,V \rangle_b \big) \, , 
\end{equation}
where $A$ is a symmetric, positive definite tensor field $A: B \rightarrow T^*B \otimes T^*B$, 
such that, for all points $b \in B$ the eigenvalues are uniformly bounded away from zero. 
\end{defn}

\begin{lem}\label{Ewindow}
The restriction of the function $\Phi$ of \eqref{PhiTB} to the bundle $TB\times \mathbb{S}(T^*B)$ 
induces a smooth function $\Psi: \cE \to \R$ from the total space of the bundle $\cE$ of signal spaces,
\begin{equation}\label{WindowFunction}
    \Psi_{(b,p)}(V)=exp\big(-V^tA_b V-i \langle \eta_p,V \rangle_b \big) \, ,
\end{equation}
where $\eta_p$ is just the point $p\in \bS^{n-1}\simeq \bS(T^*_b B)$ seen as a cotangent vector. 
\end{lem}

\proof
As discussed above, the map $h: \cE \to TB$ gives an identification of the fibers 
$\cE_{(b,p)}\simeq T_b B$, for all $(b,p)\in \pi_\bS^{-1}(b)$. Thus, we can identify tha
variables $V$ as variables in the fibers $\cE_{(b,p)}$. The choice of a point $m=(b,p)\in M$
corresponds to the choice of a point $b\in B$ and a unit vector $\eta_p$ in the unit
cotangent sphere $\bS(T^*_b B)$.
\endproof

\smallskip

\begin{defn}\label{defGaborE}
A Gabor system $\cG(\Psi,\Lambda)$ on an $n$-dimensional smooth compact manifold $B$ is 
determined by the data of a smooth window function $\Psi: \cE \to \R$ that is of rapid decay
in the fiber directions $V\in \cE_{(b,p)}$, for all $(b,p)\in M$, and a lattice bundle $\Lambda$
over $M$. The Gabor system $\cG(\Psi,\Lambda)$ with these data consists of the collection
\begin{equation}\label{GaborE}
\cG(\Psi,\Lambda)=\{ M_{\lambda'} T_{\lambda} \Psi\,|\, (\lambda, \lambda') \in \Gamma(M,\Lambda) \} \, ,
\end{equation} 
where $(\lambda, \lambda')$ are sections as in Definition~\ref{notationLambda}. 
\end{defn}

\smallskip

Lemma~\ref{lemLambda} and Lemma~\ref{Ewindow} immediately imply the following
statement.

\begin{prop}\label{GaborsysELambda}
Let $B$ be an $n$-dimensional smooth compact manifold. The collection
$\{ F_i^* \alpha \}_{i=0}^{n-1}$ of contact $1$-forms on $M=\bS(T^*B)$ 
and the function $\Phi$ of \eqref{PhiTB} determine a Gabor system 
$\cG(\Psi,\Lambda)$  on $B$, with window function $\Psi$ as in \eqref{WindowFunction}
and with lattice bundle structure $\Lambda$ on $\cE\oplus \cE^\vee$ 
as in Lemma~\ref{lemLambda} (with either the Reeb
lattice bundle structure of \eqref{LambdaEdef} or the dual-basis 
lattice bundle structure of \eqref{LambdaEdef2}).
\end{prop}

\smallskip

We can regard the Gabor system $\cG(\Psi,\Lambda)$ as a consistent
collection of Gabor systems $\cG(\Psi_m, \Lambda_m)$ in $L^2(\cE_m)$
for $m\in M$. In each fiber $\cE_m$ the wavelets in $\cG(\Psi_m, \Lambda_m)$
can be used to analyze the restriction $\cI|_{\cE_m}$ of a signal
$\cI: \cE \to \R$. 

\smallskip
\subsection{Boundary detection property}

The main reason why the construction of Gabor filters described in Proposition~\ref{GaborsysELambda}
uses the bundle $\cE$ over the manifold $M=\bS(T^*B)$ rather than the tangent bundle $TB$
over the manifold $B$ is in order to obtain filters that are especially suitable to detect 
$(n-1)$-dimensional boundaries
in a signal $f: B \to \R$ (lifted to a signal $\cI(f): \cE \to \R$ as in Lemma~\ref{MsignalEsignal}).

\smallskip

To see this property, it suffices to focus on a single fiber $\cE_m \simeq \R^n$. The
restriction of the window function $\Psi$ to this fiber is a rapid decay function of
the form \eqref{WindowFunction}, for $m=(b,p)$, where the unit cotangent vector $\eta_p$
parameterizes a choice of an oriented hyperplane in $\cE_m \simeq \R^n$. 

\smallskip

Consider a signal $f: B \to \R$ that is a characteristic function $f=\chi_U$ 
of a bounded open set $U \subset B$ with smooth boundary $\Sigma=\partial U$
given by an $(n-1)$-dimensional smooth hypersurface $\Sigma$ in $B$. 
Let $\cI(f)_{\cE_m}: \cE_m \to \R$ denote the lifted signals on the fibers
of the bundle $\cE$ of signal spaces.
As in \cite{SaCiPe}, define the output function
\begin{equation}\label{outputO}
 \cO_b(f,\eta_p):= \int_{\cE_{(b,p)}} \cI(f)|_{\cE_{(b,p)}}(V)\cdot  \Psi_{(b,p)}(V) \, \, dV\, . 
\end{equation}  

\begin{prop}\label{nBoundary}
For a given signal $f: B \to \R$ of the form $f=\chi_U$, with corresponding lift $\cI(f): \cE \to \R$, and 
for fixed $b\in \Sigma\subset B$, 
the output function $\cO_b(\eta_p)$ has a local maximum
for $p\in \bS^{n-1}$ the normal vector $\nu_b(\Sigma)$ at $b$ to the boundary 
hypersurface $\Sigma=\partial U$,
$$ {\rm argmax}^{loc}_{p\in \bS^{n-1}} \cO_b(\chi_U, \eta_p) = \nu_b(\Sigma) \, . $$
\end{prop}

\begin{proof} Let $g^*_{(b,p)}$ be the Riemannian metric on the fibers of $\mathcal{E}^\vee$, $\cE^\vee_{(b,p)}$, induced by the Riemannian metric on $B$. Then, $\mathcal{O}(\eta_p)$ has a local maximum when the gradient $\nabla^{g^*_{(b,p)}}\mathcal{O}(\eta_p)$ with respect to $g^*_{(b,p)}$ is zero. Thus, we have the following
\begin{align}\label{gradient}
\begin{split}
    \nabla^{g^*_{(b,p)}} \cO_b(\mathcal{I}(f), \eta_p) &= \int_{\cE_{(b,p)}} \cI(f)\mid_{\cE_{(b,p)}}(V)\cdot \nabla^{g^*_{(b,p)}} \Psi_{(b,p)}(V) \, \, dV  \\
    &=\int_{\cE_{(b,p)}} \chi \cdot f \circ \exp \circ h (V)\cdot i(\nabla^{g^*_{(b,p)}} \langle \eta_p,V \rangle_b) \Psi_{(b,p)}(V) \, \, dV
    \end{split}
\end{align}
At the same time, we also have  
\begin{equation*}
    I(f)(V)=\begin{cases}~1,~ \exp\circ h(V) \in ~ U\\
    ~0,~ \text{otherwise}
    \end{cases}
\end{equation*}
and therefore equation (\ref{gradient}) becomes
\begin{align}\label{gradient1}
    \nabla^{g^*_{(b,p)}} \cO_b(\mathcal{I}(f), \eta_p)= \int_{h^{-1}(T_b U)} I(f)\mid_{h^{-1}(T_b U)}(V) i(\nabla^{g^*_{(b,p)}} \langle \eta_p,V \rangle_b) \Psi_{(b,p)}(V) \, \, dV.
\end{align}
It follows, from equation (\ref{gradient1}), that $\nabla^{g^*_{(b,p)}} \cO_b(\mathcal{I}(f), \eta_p)$ is equal to zero exactly when $$\nabla^{g^*_{(b,p)}} \langle \eta_p,V \rangle_b=0, \text{ for every } V\in h^{-1}(T_b U).$$
The latter condition holds when $\langle \eta_p,V \rangle_b$ has a local maximum for every $V\in h^{-1}(T_b U)$, hence it holds when $\eta_p$ is normal to $\Sigma$, since $V\in T_b B$ at $b$.
\end{proof}

\medskip
\section{Geometric Bargmann transforms and Gabor frames}

In order to ensure that a Gabor system $\cG(\Psi,\Lambda)$ on an $n$-dimensional manifold $B$,
in the sense of Definition~\ref{defGaborE} has good signal analysis properties, we need a method to
detect whether it satisfies the frame condition. 

\smallskip

The local frames $\{V_1,...,V_n\}$ and $\{ \alpha,F_1^* \alpha, \ldots ,F^*_{n-1}\alpha\}$ of $\mathcal{E}$ and $\mathcal{E}^\vee$ determine a 
local isomorphism between $\mathcal{E}$ and $\mathcal{E}^\vee$. 
For $(V,\eta)$ in the fiber $(\mathcal{E}\oplus\mathcal{E}^\vee)_{(\bar{x},\bar{y})}$,
we define the vector bundle morphism 

\begin{align}\label{ComplexiIso}
\begin{split}
\mathcal{\mathcal{I}}: \mathcal{E}\oplus\mathcal{E}^\vee\rightarrow \mathcal{E}\oplus\mathcal{E}^\vee\, ,\
(V,\eta)\mapsto (\eta,-V):= \sum_{i=1}^n\eta_i V_i -v_i F_{i-1}^* \alpha
\end{split}
    \end{align}
    which satisfies the condition $\mathcal{I}^2=-1$ and gives a $\mathbb{C}$-linear isomorphism 
   \begin{align*}
      \mathcal{J}_{\bar{x},\bar{y}}:\mathcal{E}\oplus\mathcal{E}^\vee_{\bar{x},\bar{y}}&\rightarrow \mathbb{C}^{n}  \\
      (V,\eta) &\mapsto (v_1+i\eta_1,...,v_n+i\eta_n).
       \end{align*}
        with scalar multiplication by $\lambda\in \C$, $\lambda=x+iy$
with $x,y\in \R$ given by $\lambda \cdot (V+ i \eta)=(x+y\, \cI)\,  (V,\eta)$.
    \begin{defn}\label{Bargmann}
    The Bargmann Transform of a function $f\in L^2(\mathcal{E},\mathbb{C})$ is the function $\mathcal{B}f:\mathcal{E}\oplus \mathcal{E}^\vee \rightarrow \mathbb{C}$ defined at each fiber as
    \begin{equation}
        \mathcal{B}f_{\mid_{\mathcal{E}_{\bar{x},\bar{y}}}(V,\eta)}:=\int_{\mathcal{E}_{\bar{x},\bar{y}}}f_{\mid _{\mathcal{E}_{\bar{x},\bar{y}}}}(W)e^{2\pi W\star (V,\eta)- W^tA_{\bar{x}}W-\frac{\pi}{2}\mathcal{P}(V,\eta)}dvol_{\bar{x},\bar{y}}(W),
    \end{equation}
    where $$W\star (V,\eta):=W^t\frac{A_{\bar{x}}}{\pi}V+i\langle \eta,V\rangle$$
    and $\mathcal{P}:\mathcal{E}\oplus\mathcal{E}^\vee\rightarrow \mathbb{C}$ the quadratic form associated to $A$ defined as
    $$\mathcal{P}(V,\eta):= V^t\frac{A_{\bar{x}}}{\pi}V+2i\langle \eta,V\rangle-\eta^t\eta.$$
    The volume form $dvol_{\bar{x},\bar{y}}(W)$ is the volume form on the fibers of $\mathcal{E}$ determined by the Riemannian metric on M. 
    \end{defn}
    \begin{defn}Let $z$ denote $\mathcal{J}(V,\eta+\frac{\eta_{\bar{y}}}{2\pi})$ for some $(V,\eta)$ in $\mathcal{E}\oplus\mathcal{E}^\vee$, then

    \begin{itemize}
    
        \item The Bargmann-Fock space $\mathcal{F}^2(\mathcal{E}\oplus \mathcal{E}^\vee)$ is the space of functions $f: \mathcal{E}\oplus \mathcal{E}^\vee\rightarrow \mathbb{C}$ such that the $f_{\mid _{\mathcal{E}_{(\bar{x},\bar{y})}}}\circ \mathcal{J}^{-1}: \mathbb{C}^n\rightarrow \mathbb{C}$ is entire and they are bounded with respect to the norm 
    \begin{align}
\|f&\|_{\mathcal{F}^2(\mathcal{E}\oplus\mathcal{E}^\vee)}:=\nonumber \\
&\int_M\Big (\int_{\mathbb{C}^n} \mid f_{\mid_{(\mathcal{E}\oplus\mathcal{E}^\vee)_{\bar{x},\bar{y}}}}\circ \mathcal{J}^{-1}(z) \mid ^2e^{-\frac{\pi}{2}(\mathcal{R}e(z)^t\frac{A_{\bar{x}}}{\pi}\mathcal{R}e(z)+\mathcal{I}m(z)^t\mathcal{I}m(z))}dz\Big )^\frac{1}{2}dvol_M
    \end{align}
     \item The fiberwise Bargmann-Fock space 
     $\mathcal{F}^2(\mathcal{E}\oplus\mathcal{E}^\vee)_{\bar{x},\bar{y}}$, is the space of all functions \\ $F:(\mathcal{E}\oplus\mathcal{E}^\vee)_{\bar{x},\bar{y}}\rightarrow \mathbb{C}$ such that $F\circ \mathcal{J}^{-1}:\mathbb{C}^n\rightarrow \mathbb{C}$ is entire and  
     \begin{align}
    \|F&\|_{\mathcal{F}^n((\mathcal{E}\oplus\mathcal{E}^\vee)_{\bar{x}\bar{y}})}:=\nonumber \\
    &\Big (\int_{\mathbb{C}^n} \mid F_{\mid_{(\mathcal{E}\oplus\mathcal{E}^\vee)_{\bar{x},\bar{y}}}}\circ \mathcal{J}^{-1}(z)\mid ^2e^{-\frac{\pi}{2}(\mathcal{R}e(z)^t\frac{A_{\bar{x}}}{\pi}\mathcal{R}e(z)+\mathcal{I}m(z)^t\mathcal{I}m(z))}dz\Big )^{\frac{1}{2}}<\infty
    \end{align}
    \item A subset $\Lambda$ of $\mathbb{C}^n$ is a set of sampling for $\mathcal{F}^2(\mathcal{E}\oplus\mathcal{E}^\vee)_{\bar{x},\bar{y}}$ if there exist $A$ and $B$ smooth, positive functions on the local charts of $M$ such that  
    \begin{align*}
        A_{\bar{x},\bar{y}}\|F\|_{\mathcal{F}^2(\mathcal{E}\oplus\mathcal{E}^\vee_{\bar{x},\bar{y}})}\leq & \\
        \sum_{\lambda\in \Lambda }\mid F\mid _{\mathcal{E}\oplus\mathcal{E}_{\bar{x},\bar{y}}}&\circ 
        \mathcal{J}^{-1}(\lambda)\mid ^2e^{-\frac{\pi}{2}(\mathcal{R}e(\lambda)^t\frac{A_{\bar{x}}}{\pi}\mathcal{R}e(\lambda)+\mathcal{I}m(\lambda)^t\mathcal{I}m(\lambda))}\leq\\
    &\hspace{4cm}   B_{\bar{x},\bar{y}
       }\|F\|_{\mathcal{F}^2(\mathcal{E}\oplus\mathcal{E}^\vee_{\bar{x},\bar{y}
       })
       }
    \end{align*}
    \end{itemize}
    \begin{prop}\label{BargDirInt}
    The Bargmann-Fock space $\mathcal{F}^2(\mathcal{E}\oplus\mathcal{E}^\vee)$ with inner product 
    \begin{align*}
        \langle F,G &\rangle_{\mathcal{F}^2(\mathcal{E}\oplus\mathcal{E}^\vee)}:= \\
        &\int_M\Big (\int_{\mathbb{C}^n}F\circ\mathcal{J}^{-1}(z)\overline{G\circ\mathcal{J}^{-1}(z)}e^{-\frac{\pi}{2}(\mathcal{R}e(z)^t\frac{A_{\bar{x}}}{\pi}\mathcal{R}e(z)+\mathcal{I}m(z)^t\mathcal{I}m(z))}dz\Big )dvol_M
    \end{align*}
    and the fiberwise Bargmann-Fock space $\mathcal{F}^2(\mathcal{E}\oplus\mathcal{E}^\vee)_{\bar{x},\bar{y}}$ with inner product 
    
    \begin{equation}
    \langle F,G\rangle_{\mathcal{F}^2(\mathcal{E}\oplus\mathcal{E}^\vee_{\bar{x},\bar{y}})}:=\int_{\mathbb{C}^n}F\circ\mathcal{J}^{-1}(z)\overline{G\circ\mathcal{J}^{-1}(z)}e^{-\frac{\pi}{2}(\mathcal{R}e(z)^t\frac{A_{\bar{x}}}{\pi}\mathcal{R}e(z)+\mathcal{I}m(z)^t\mathcal{I}m(z))}dz
    \end{equation}
    
    are Hilbert spaces. Additionally, the mapping 
    \begin{align}
        \mathcal{F}^2(\mathcal{E}\oplus\mathcal{E}^\vee_{\bar{x},\bar{y}})\rightarrow \mathcal{F}^2(\mathbb{C}^n)\, , \
        F\mapsto F\circ\mathcal{J}^{-1}
        \end{align}
        is an embedding if and only if 
        \begin{equation*}
            \rho(Q_{\bar{x}})\leq 1
        \end{equation*}
        where $Q_{\bar{x}}$ is the inveritible matrix such that $A_{\bar{x}}=Q^t_{\bar{x}}Q_{\bar{x}}$ and $\rho(Q_{\bar{x}})$ the spectral radius of $A_{\bar{x}}$.
    \end{prop}
    \begin{proof}
 
      To prove that $\cF^2(\cE\oplus\cE_{\bar{x},\bar{y}})$ is a Hilbert space, it suffices to prove that $\mathcal{F}^2(\mathcal{E}\oplus\mathcal{E}^\vee_{
      \bar{x},\bar{y}
      })$ is complete with respect to the norm
      \begin{equation*}
          \|F\|_{\mathcal{F}^2(\mathcal{E}\oplus\mathcal{E}^\vee_{
      \bar{x},\bar{y}
      })}=\sqrt{ \langle F,F\rangle_{\mathcal{F}^2(\mathcal{E}\oplus\mathcal{E}^\vee_{
      \bar{x},\bar{y}
      })}}.
      \end{equation*}
      Let $\{F_n\}_{n\in\mathbb{N}}\subset \mathcal{F}^2(\mathcal{E}\oplus\mathcal{E}^\vee_{
      \bar{x},\bar{y}
      })$ be a Cauchy sequence, then $\{F_n\circ\mathcal{J}^{-1}\circ P_{\bar{x}}^{-1}\}_{n\in \mathbb{N}}$ is a Cauchy sequence in $\mathcal{F}^2(\mathbb{C}^n)$, where $P_{\bar{x}}=\begin{pmatrix}
          \frac{Q_{\bar{x}}}{\sqrt{\pi}}&0\\
          0& I
      \end{pmatrix}$.\\
      Let, also, $F\in \mathcal{F}^2(\mathbb{C}^n)$ such that $||F_n\circ\mathcal{J}^{-1}\circ P_{\bar{x}}^{-1}-F||_{\mathcal{F}^2(\mathbb{C}^n)}\rightarrow 0$. Then, 
      \begin{align*}
         \|F_n\circ\mathcal{J}^{-1}\circ P_{\bar{x}}^{-1}-F\|_{\mathcal{F}^2(\mathbb{C}^n)}&= \int_{\mathbb{C}^n}\mid F_n\circ\mathcal{J}^{-1}\circ P_{\bar{x}}^{-1}-F\mid ^2e^{-\pi|z|^2}dz \\
         &=det(P_{\bar{x}}) \int_{\mathbb{C}^n}\mid F_n\circ \mathcal{J}^{-1}-F\circ P_{\bar{x}}\mid e^{-\pi \mid Pz\mid ^2}dz\\
         &= det(P_{\bar{x}})\|F_n-F\circ\mathcal{J}\|_{\mathcal{F}^2(\mathcal{E}\oplus\mathcal{E}^\vee_{
      \bar{x},\bar{y}.
      })}
      \end{align*}
      
      Thus, $\{F_n\circ\mathcal{J}^{-1}\circ P_{\bar{x}}^{-1}\}_{n\in \mathbb{N}}$ converges to the function $F_{P_{\bar{x}}}\circ \mathcal{J}\in \mathcal{F}^2(\mathcal{E}\oplus\mathcal{E}^\vee_{
      \bar{x},\bar{y}
      })$, where $F_{P_{\bar{x}}}=F\circ P_{\bar{x}}$, with respect to the $\mathcal{F}^2(\mathcal{E}\oplus\mathcal{E}^\vee_{
      \bar{x},\bar{y}
      })$ norm.

      Moreover, the space $\mathcal{F}^2(\mathcal{E}\oplus\mathcal{E}^\vee)$ is the direct integral of the Hilbert spaces $\mathcal{F}^2(\mathcal{E}\oplus\mathcal{E}^\vee_{
      \bar{x},\bar{y}
      })$. The inner product 
      \begin{equation*}
          \langle F,G\rangle =\int_M \langle F,G\rangle_{\mathcal{F}^2(\mathcal{E}\oplus\mathcal{E}^\vee_{
      \bar{x},\bar{y}
      })}dvol_M
      \end{equation*}
      turns $\mathcal{F}^2(\mathcal{E}\oplus\mathcal{E}^\vee)$ into a Hilbert space.

      Finally, to prove that the map $F\mapsto F\circ \cJ^{-1}$ is an embedding of $\cF^2((\cE\oplus \cE^\vee)_{\bar{x},\bar{y}})$ into $\cF^2(\bC^n)$, let $F\in \mathcal{F}^2(\mathcal{E}\oplus\mathcal{E}^\vee_{
      \bar{x},\bar{y}
      })$. To show that $F\circ\mathcal{J}^{-1}$ is in $\mathcal{F}^2(\mathbb{C}^n)$, it suffices to show that $\mathcal{F}^2(\mathbb{C}^n, e^{-\frac{\pi}{2}\mid Pz\mid ^2}dz)$ is embedded in $\mathcal{F}^2(\mathbb{C}^n)$. The latter holds exactly when \\
      $\mid Pz\mid \leq \mid z\mid$, which, in turn, is satisfied exactly when $\rho(Q_{\bar{x}})\leq 1$.
      Since 
      \begin{equation*}
          \mid Pz\mid ^2\leq \rho(Q_{\bar{x}})^2\mid Re(z)\mid ^2+\mid Im(z)\mid ^2\leq \mid z\mid ^2 \text{ and } \rho(Q_{\bar{x}})\leq \|Q_{\bar{x}}\|.
      \end{equation*}

    \end{proof}
    \end{defn}
     \begin{rem} {\rm
         
      For $(W,\xi)\in \mathcal{F}^2(\mathcal{E}\oplus \mathcal{E}^\vee)_{\bar{x},\bar{y}}$,
        we introduce the notation 
        $$z:= \mathcal{J}(W,\xi)\text{
        and } 
\mathfrak{P}(z):=\mathcal{R}e(z)^t\frac{A_{\Bar{x}}}{\pi}\mathcal{R}e(z)+\mathcal{I}m(z)^t\mathcal{I}m(z)$$ }
\end{rem}
\begin{lem}\label{norm}
   The elementary Gabor functions $M_\xi T_W\Psi_{\bar{x},\bar{y}}:\mathcal{E}_{\bar{x},\bar{y}}\rightarrow\mathbb{C}$ satisfy the condition
   \begin{equation*}
       \mid\langle f,M_\xi T_W\Psi_{\bar{x},\bar{y}}\rangle\mid=e^{-\frac{\pi}{2}\mathfrak{P}(z)}\mid\mathcal{B}f(\bar{z})\mid
   \end{equation*}
   for every $f$ in $L^2(\mathcal{E}_{\bar{x},\bar{y}})$ and every $(\xi,W)$ in $  (\mathcal{E}\oplus\mathcal{E}^\vee)_{\bar{x},\bar{y}}$.
   
    \end{lem}
     \begin{proof}
    Let $f\in L^2(\mathcal{E}_{\bar{x},\bar{y}})$ and $(\xi,W)\in (\mathcal{E}\oplus\mathcal{E}^\vee)_{\bar{x},\bar{y}}$, then
    
    \begin{align}\label{STFT}
        \langle f,M_{\xi} T_W\Psi_{\bar{x},\bar{y}}\rangle &=\int_{\mathcal{E}_{\bar{x},\bar{y}
        }
        }
        f(V)e^{(V-W)^tA_{\bar{x}}(V-W)}e^{i\langle\eta_{\bar{y}},V-W\rangle -2\pi i \langle\xi, V\rangle}
        dvol_{\mathcal{E}_{\bar{x},\bar{y}}}(V) \\   
        &=e^{-i\langle\eta_{\bar{y}},W\rangle}e^{-\pi(W^t\frac{A_{\bar{x}}}{\pi}W)}\\
        \int_{\mathcal{E}_{\bar{x},\bar{y}}}f(V) & e^{-\pi (V^t\frac{A_{\bar{x}}}{\pi}V)+2\pi(V^t\frac{A_{\bar{x}}}{\pi}W)}e^{i\langle\eta_{\bar{y}},V\rangle-2\pi i \langle\eta,V\rangle} dvol_{\mathcal{E}_{\bar{x},\bar{y}}}(V)\notag \\
       & =e^{-i\langle\eta_{\bar{y}},W\rangle}e^{-i\pi\langle\xi-\frac{\eta_{\bar{y}}
        }{2\pi},W\rangle} \cdot
        e^{-\frac{\pi}{2}(W^t\frac{A_{\bar{x}}}{\pi}W+(\xi-\frac{\eta_{\bar{y}}
        }{2\pi})^t(\xi-\frac{\eta_{\bar{y}}}{2\pi}))
        } \notag\\ 
         \int_{\mathcal{E}_{\bar{x},\bar{y}}}f(V) & e^{V^t A_{\bar{x}}V+2\pi V\star (W,-\xi+\frac{\eta_{\bar{y}}}{2\pi})
        } e^{-\frac{\pi}{2}\mathcal{P}(V,-\xi+\frac{\xi_{\bar{y}}
        }{2\pi})}dvol_{\mathcal{E}_{\bar{x},\bar{y}
        }
        }(V).  
    \end{align}
    Thus, the following relation is satisfied
    \begin{equation*}
        \mid \langle f,M_\xi T_W\Psi_{\bar{x},\bar{y}}\rangle\mid
        =e^{
        -\frac{\pi}{2}\big[(W^t\frac{A_{\bar{x}}}{\pi}W)+(\xi-\frac{\eta_{\bar{y}}
        }{2\pi})^t(\xi-\frac{\eta_{\bar{y}}}{2\pi})\big]}\mid\mathcal{B}f(W,-(\xi-\frac{\eta_{\bar{y}}}{2\pi}))\mid.
    \end{equation*}
    \end{proof}
    \begin{prop} Let $\mathcal{P}:\mathcal{E}\oplus\mathcal{E}^\vee\rightarrow \mathbb{C}$ be the quadratic form associated to $A$ defined as
    $$\mathcal{P}(V,\eta):= V^t\frac{A_{\bar{x}}}{\pi}V+2i\langle \eta,V\rangle-\eta^t\eta,$$ then the following hold:
        \begin{enumerate}
            \item The vectors 
            $ e_\alpha(z)=det(\mathcal{P})\Big(\frac{\pi^{\mid\alpha\mid}}{\alpha!}\Big)^{1/2}(\mathcal{J}^{-1}_{(\bar{x},\bar{y})}(\mathcal{P}z))^\alpha
           $
            for $\alpha=(\alpha_1,...,\alpha_n)$ with $\alpha_j\geq 0$, form an orthonormal basis for $\mathcal{F}^2(\mathcal{E}\oplus\mathcal{E}^\vee)_{\bar{x},\bar{y}}$. 
            \item The Hilbert space $\mathcal{F}^2(\mathcal{E}\oplus \mathcal{E}^\vee)_{\bar{x},\bar{y}}$ is a reproducing kernel Hilbert space, that is, for every $z_0\in \mathbb{C}^n$ the evaluation functional
            $f\mapsto f(z_0)$ is bounded. The reproducing kernel is $$F_{z_0}(z)=det^2(\mathcal{P})e^{\pi\overline{(\mathcal{P}z_0)}\mathcal{P}z}$$
        \end{enumerate}
    \end{prop}
    \begin{proof}
        \begin{enumerate}
            \item The inner product of the vectors $e_\alpha(z)$ in $\mathcal{F}^2(\mathcal{E}\oplus \mathcal{E}^\vee)_{\bar{x},\bar{y}}$ can be expressed in terms of the inner product of the vectors 
            $$E_\alpha(z)=\Big(\frac{\pi^{\mid\alpha\mid}}{\alpha!}\Big)^{1/2}(z)^\alpha, z\in \mathbb{C}^n$$ in the Bargmann-Fock space $\mathcal{F}^2(\mathbb{C}^n)$ as 
            $$\langle e_\alpha(z),e_{\beta}(z)\rangle_{\mathcal{F}^2(\mathcal{E}\oplus\mathcal{E}^\vee)_{\bar{x},\bar{y}}}=\det(P^{-1})\langle E_{\alpha}(z),E_\beta(z)\rangle_{\mathcal{F}(\mathbb{C}^n)}.$$
            Since $\{E_\alpha(z): a
            \geq 0\}$ is an orthonormal basis of $\mathcal{F}^2(\mathbb{C}^n)$, \cite{Bargmann}, it follows that the vectors $e_\alpha(z)$ form an orthonormal system of $\mathcal{F}^2(\mathcal{E}\oplus\mathcal{E}^\vee)_{\bar{x},\bar{y}}$. To prove completeness, we take $f\in \mathcal{F}^2(\mathcal{E}\oplus\mathcal{E}^\vee)_{\bar{x},\bar{y}}$ such that $\langle f,e_\alpha(z)\rangle_{\mathcal{F}^2(\mathcal{E}\oplus\mathcal{E}^\vee)_{\bar{x},\bar{y}}}=0$ for all $\alpha$. Since $\langle f,e_\alpha(z)\rangle_{\mathcal{F}^2(\mathcal{E}\oplus\mathcal{E}^\vee)_{\bar{x},\bar{y}}}=\det(P^{-1})\langle f\circ \mathcal{J}^{-1}\circ P^{-1},E_\beta(z)\rangle_{\mathcal{F}(\mathbb{C}^n)}$ and $\{E_\alpha(z): a
            \geq 0\}$ is an orthonormal basis of $\mathcal{F}^2(\mathbb{C}^n)$, $f\equiv 0$.
            \item Since $f(z)=\sum_{\alpha\geq 0}\langle f,e_\alpha\rangle_{\mathcal{F}^2(\mathcal{E}\oplus\mathcal{E}^\vee)_{\bar{x},\bar{y}}}e_\alpha(z)$, by the Cauchy-Schwarz inequality we obtain
            \begin{align*}\mid f(z)\mid &\leq \big(\sum_{\alpha\geq 0}\mid \langle f,e_\alpha\rangle _{\mathcal{F}^2(\mathcal{E}\oplus\mathcal{E}^\vee)_{\bar{x},\bar{y}}}\mid ^2\big)^{1/2}\big(\sum_{\alpha\geq 0}det^2(\mathcal{P})\frac{\pi^\alpha}{\alpha !}\mid \mathcal{P}z\mid ^2\big)^{1/2}\\
         &\hspace{6cm}   \leq det^2(\mathcal{P})\|f\|_{\mathcal{F}^2(\mathcal{E}\oplus\mathcal{E}^\vee)_{\bar{x},\bar{y}}}\cdot e^{\pi\mid\mathcal{P}z\mid ^2/2}. \end{align*}
            Therefore, the point evaluations for $z\in \mathbb{C}^n$ 
            \begin{align*}
                \mathcal{F}^2(\mathcal{E}\oplus\mathcal{E}^\vee)_{\bar{x},\bar{y}}&\rightarrow \mathbb{C}\\
                f&\mapsto f(z)
            \end{align*}
            are continuous linear functionals and for each $z_0\in \mathbb{C}^n$ there exists a reproducing kernel $F_{z_0}\in \mathcal{F}^2(\mathcal{E}\oplus\mathcal{E}^\vee)_{\bar{x},\bar{y}}$ such that
            \begin{equation}
                f(z_0)=\langle f,F_{z_0}\rangle.
            \end{equation}
          Expanding the reproducing kernel $F_{z_0}$ with respect to the orthonormal basis $\{e_\alpha\}_{\alpha\geq 0}$, we obtain the following equation
            \begin{align*}
                F_{z_0}(z)&=\sum_{\alpha}\langle F_{z_0},e_\alpha\rangle e_\alpha(z)\\
                &=\sum_{\alpha} \overline{e_\alpha(z_0)}e_\alpha(z)  \\
                &= \sum_{\alpha} \det^2(\mathcal{P})\frac{\pi^{\mid\alpha\mid}}{\alpha !}\overline{(\mathcal{P}z_0)^\alpha}(\mathcal{P}z)^\alpha\\
                &=\det^2(\mathcal{P})e^{\pi\overline{(\mathcal{P}z_0)}\mathcal{P}z}
            \end{align*}
            and that completes the proof.
        \end{enumerate}
    \end{proof}
    \begin{prop}

    The Bargmann Transform of definition \ref{Bargmann} is a linear bijection from $L^2(\mathcal{E}_{(\bar{x},\bar{y})})$ onto $\mathcal{F}^n(\mathcal{E}_{(\bar{x},\bar{y})}\oplus\mathcal{E}^\vee_{(\bar{x},\bar{y})})$, for any $(\bar{x},\bar{y})\in M$, and 
    \begin{equation}
        \|\mathcal{B}f\|_{\mathcal{F}(\mathcal{E}_{(\bar{x},\bar{y})}\oplus\mathcal{E}^\vee_{(\bar{x},\bar{y})})}=\sqrt{\frac{\pi^n}{detA_{\bar{x}}}}\|f\|_{L^2(\mathcal{E}_{\bar{x}})}.
    \end{equation}
    \end{prop}
    \begin{proof}
    Let $(\bar{x},\bar{y})$ in $M$, $(\eta,V)$ in $ (\mathcal{E}\oplus\mathcal{E}^\vee)_{\bar{x},\bar{y}}$ and $f$ in $L^2(\mathcal{E}_{\bar{x},\bar{y}})$. 
From the orthogonality relations of the Short Time Fourier Transform it follows that 
\begin{equation*}
 \|\langle f,M_\eta T_V\Psi_{\bar{x},
    \bar{y}}\rangle\|^2_{L^2(\mathcal{E}_{\bar{x},\bar{y}})}=\|f\|^2_{L(\mathcal{E}_{\bar{x},\bar{y}})}\|\Psi_{\bar{x},\bar{y}}\|^2_{L^2(\mathcal{E}_{\bar{x},\bar{y}})}.
\end{equation*}
Additionally, for the norm $\|\Psi_{\bar{x},\bar{y}}\|^2_{L^2(\mathcal{E}_{\bar{x},\bar{y}})}$ of the window function holds that
\begin{equation*}
    \|\Psi_{\bar{x},\bar{y}}\|^2_{L^2(\mathcal{E}_{\bar{x},\bar{y}})}=\int_{\mathcal{E}_{\bar{x},\bar{y}}}e^{-WA_{\bar{x}}W}dvol_{\mathcal{E}_{\bar{x},\bar{y}}}(W)=\sqrt{\frac{\pi^n}{det A_{\bar{x}}}},
    \end{equation*}
and by Lemma \ref{norm} it follows that the fiberwise Bargmann Transform is bounded, injective and 
\begin{equation*}
    \|\mathcal{B}f\|_{\mathcal{F}^2(\mathcal{E}_{(\bar{x},\bar{y})}\oplus\mathcal{E}^\vee_{(\bar{x},\bar{y})})}^2=\sqrt{\frac{\pi^n}{det A_{\bar{x}}}}\|f\|^2_{L^2(\mathcal{E}_{\bar{x},\bar{y}})}.
\end{equation*} 
To prove surjectivity of $\mathcal{B}$ onto $\mathcal{F}^2(\mathcal{E}_{(\bar{x},\bar{y})}\oplus\mathcal{E}^\vee_{(\bar{x},\bar{y})})$, it suffices to prove that $\mathcal{B}(L^2(\mathcal{E}_{\bar{x},\bar{y}}))$ is dense in $\mathcal{F}^2(\mathcal{E}_{(\bar{x},\bar{y})}\oplus\mathcal{E}^\vee_{(\bar{x},\bar{y})})$. We write $z=\mathcal{J}(W,\xi)$ and $z_0=\mathcal{J}(V,\eta)$. After some bookkeeping and after applying equation (\ref{STFT}) , we obtain that
\begin{equation*}
    \mathcal{B}(T_V M_{\eta} \Psi_{\bar{x},\bar{y}})(W,-\xi)= h(V, \eta, \frac{\eta_{\bar{y}}}{2\pi})det(\mathcal{P})e^{\pi\overline{\mathcal{P}z_0}\mathcal{P}z},
\end{equation*}
for some $h\in \mathbb{C}$ depending on $V,\eta$ and $\frac{\eta_{\bar{y}}}{2\pi}$. Thus, the reproducing kernel $K_{z_0}(z)$ is in the range of $\mathcal{B}$ Suppose that there exists some $F\in \mathcal{F}^2((\mathcal{E}\oplus\mathcal{E}^\vee)_{\bar{x},\bar{y}})$ such that $\langle F,\mathcal{B}f\rangle =0$, for all $f\in L^2(\mathcal{E}_{\bar{x},\bar{y}})$. Then, for every $z_0=\mathcal{J}(V,\eta)$
\begin{equation*}
    0=\langle\mathcal{B}(T_V M_\eta \Psi_{\bar{x},\bar{y}}),F\rangle=h(V, \eta, \frac{\eta_{\bar{y}}}{2\pi})det(\mathcal{P})\langle K_{z_0}z,F\rangle =F(z_0) 
\end{equation*}
and therefore $F\equiv 0$. Thus, $\mathcal{B}(L^2(\mathcal{E}_{
\bar{x},\bar{y}
}))=\mathcal{F}^2((\mathcal{E}\oplus\mathcal{E}^\vee)_{\bar{x},\bar{y}})$
    \end{proof}
   Finally, by considering a complex lattice as in \cite{Groch2} we have the following statement. 
\begin{thm}
Let $\mathcal{G}(\Psi_{(\bar{x},\bar{y})},\Lambda_{\bar{x},\bar{y}})$ a Gabor system with window function as defined in (\ref{GaborE}). Suppose that there exist  normalized lattices $L_1,...,L_n$ in $\mathbb{C}$, $M$ in 
$GL(n,\mathbb{C})$ and $a\in \mathbb{C}^*$ such that  
$$\mathcal{J}(\Lambda_{\bar{x},\bar{y}})=aM
\bigoplus_{i=1}^{n}L_i$$
and for the characteristic indices $\gamma_1,...,\gamma_n$ of $$\begin{pmatrix}\frac{Q_{\bar{x}}}{\sqrt{\pi}} & 0\\
0& I\end{pmatrix}M$$ it holds that $0<\gamma_i<1$ for $i=1,...,n$, where $Q_{\bar{x}}$ is the invertible matrix such that $A_{\bar{x}}=Q_{\bar{x}}^TQ_{\bar{x}}$, then $\mathcal{G}(\Psi_{(\bar{x},\bar{y})},\Lambda_{\bar{x},\bar{y}})$ is a frame.

\end{thm}

\begin{proof}
From Lemma \ref{norm} it follows that $\mathcal{G}(\Psi_{\bar{x},\bar{y}},\Lambda_{\bar{x},\bar{y}})$ is a frame exactly when $\overline{\Lambda_{\bar{x},\bar{y}}}-i\frac{\eta_{\bar{y}}}{2\pi}$ is a set of sampling for the Bargamann-Fock space $\mathcal{F}^2((\mathcal{E}\oplus\mathcal{E}^\vee)_{\bar{x},\bar{y}})$. Hence, it suffices to prove that there exist $A:U\rightarrow \mathbb{R}^+$  and $B:U\rightarrow \mathbb{R}^+$ smooth functions in the local chart $U\subseteq M$ that contains $(\bar{x},\bar{y})$ such that 
\begin{alignat*}{2}
        A_{\bar{x},\bar{y}}
        &\|F\|^2_{\mathcal{F}^2(\mathcal{E}\oplus\mathcal{E}^\vee_{\bar{x},\bar{y}})}\leq \hspace{2cm}&&\hspace{1cm}
        \\ &\sum_{\lambda\in \overline{\Lambda}_{\bar{x},\bar{y}}-i\frac{\eta_{\bar{y}}}{2\pi} }
        \mid F\mid _{\mathcal{E}\oplus\mathcal{E}_{\bar{x},\bar{y}}}
        \circ \mathcal{J}^{-1}(\lambda)\mid ^2e^{-\frac{\pi}{2}(\mathcal{R}e(\lambda)^t\frac{A_{\bar{x}}}{\pi}
        \mathcal{R}e(\lambda)+\mathcal{I}m(\lambda)^t\mathcal{I}m(\lambda))}&&\leq\\
        \hspace{1cm}&\hspace{1cm}&& B_{\bar{x},\bar{y}}
    \|F\|^2_{\mathcal{F}^2(\mathcal{E}\oplus\mathcal{E}^\vee_{\bar{x},\bar{y}})}
    \end{alignat*}
for every $F\in \mathcal{F}^2((\mathcal{E}\oplus\mathcal{E}^\vee)_{\bar{x},\bar{y}})$ or equivalently 

\begin{alignat*}{2}
        A_{\bar{x},\bar{y}}
        &\|F\|^2_{\mathcal{F}^2(\mathcal{E}\oplus\mathcal{E}^\vee_{\bar{x},\bar{y}})}\leq \hspace{2cm}&&\hspace{1cm}
        \\ &\sum_{\lambda\in \overline{\widetilde{\Lambda}}_{\bar{x},\bar{y}}-i\frac{\eta_{\bar{y}}}{2\pi} }
        \mid F\mid _{\mathcal{E}\oplus\mathcal{E}_{\bar{x},\bar{y}}}
        \circ \mathcal{J}^{-1}(P^{-1}\lambda)\mid ^2e^{-\frac{\pi}{2}(\mathcal{R}e(\lambda)^t
        \mathcal{R}e(\lambda)+\mathcal{I}m(\lambda)^t\mathcal{I}m(\lambda))}&&\leq\hspace{1cm}\\
        \hspace{1cm}& && B_{\bar{x},\bar{y}}
     \|F\|^2_{\mathcal{F}^2(\mathcal{E}\oplus\mathcal{E}^\vee _{\bar{x},\bar{y}})} 
    \end{alignat*}
where $P=\begin{pmatrix}\frac{Q_{\bar{x}}}{\sqrt{\pi}}& 0\\
0& I\end{pmatrix}
$
and $\widetilde{\Lambda_{\bar{x},\bar{y}}}=P\Lambda_{\bar{x},\bar{y}}$. By Theorem 9 of \cite{Groch} and Proposition 4.5 of \cite{LiMa}, it follows that $\widetilde{\Lambda_{\bar{x},\bar{y}}}-i\frac{\eta_{\bar{y}}}{2\pi}$ is a sampling set of $\mathcal{F}^2(\mathbb{C}^n)$ since for the characteristic indices $\{\gamma_i, i=1,..,n\}$ of  $PM$ it holds that $0<\gamma_i<1$. Additionally, $F\mid_{\mathcal{E}\oplus\mathcal{E}^{\vee}_{\bar{x},\bar{y}}}\circ \mathcal{J}^{-1}$ is in $\mathcal{F}^2(\mathbb{C}^n, e^{-\frac{\pi}{2}\mid Pz\mid ^2}dz)$ since $F\in \mathcal{F}^2(\mathcal{E}\oplus\mathcal{E}^{\vee}_{\bar{x},\bar{y}})$ and $\|F\circ \mathcal{J}^{-1}\circ P^{-1}\|^2_{\mathcal{F}^2(\mathbb{C}^n)}=det(P)^2\|F\circ \mathcal{J}^{-1}\|^2_{\mathcal{F}^2(\mathbb{C}^n, e^{-\frac{\pi}{2}\mid Pz\mid ^2}dz)}$. Thus the following inequality holds
\begin{alignat*}{2}
        A_{\bar{x},\bar{y}}det(P)^2
        &\|F\circ\mathcal{J}^{-1}\|^2_{\mathcal{F}^2(\mathbb{C}^n, e^{-\frac{\pi}{2}\mid Pz\mid ^2}dz)}\leq &&
        \\ &\sum_{\lambda\in \overline{\widetilde{\Lambda}}_{\bar{x},\bar{y}}-i\frac{\eta_{\bar{y}}}{2\pi} }
        \mid F\mid _{\mathcal{E}\oplus\mathcal{E}_{\bar{x},\bar{y}}}
        \circ \mathcal{J}^{-1}(P^{-1}\lambda)\mid ^2&&e^{-\frac{\pi}{2}(\mathcal{R}e(\lambda)^t
        \mathcal{R}e(\lambda)+\mathcal{I}m(\lambda)^t\mathcal{I}m(\lambda))}\\
        &  \hspace{-5cm} &&
        \leq B_{\bar{x},\bar{y}}\det(P)^2
        \|F\circ\mathcal{J}^{-1}\| ^2_{\mathcal{F}^2(\mathbb{C}^n, e^{-\frac{\pi}{2}\mid Pz\mid ^2}dz)}.
 \end{alignat*}
The latter is equivalent to the inequality 
\begin{alignat*}{2}
        A_{\bar{x},\bar{y}}& det(P)^2
        \|F\|^2_{\mathcal{F}^2(\mathcal{E}\oplus\mathcal{E}^\vee_{\bar{x},\bar{y}})}\leq && 
        \\ &\sum_{\lambda\in \overline{{\Lambda}}_{\bar{x},\bar{y}}-i\frac{\eta_{\bar{y}}}{2\pi} }
        \mid F_{\mid _{\mathcal{E}\oplus\mathcal{E}_{\bar{x},\bar{y}}}}
        \circ \mathcal{J}^{-1}(\lambda)\mid && ^2e^{-\frac{\pi}{2}(\mathcal{R}e(\lambda)^t\frac{A_{\bar{x}}}{\pi}
        \mathcal{R}e(\lambda)+\mathcal{I}m(\lambda)^t\mathcal{I}m(\lambda))}\leq\hspace{3cm}\\
        \hspace{1cm}&\hspace{1cm} &&\hspace{3cm}B_{\bar{x},\bar{y}}\det(P)^2
\|F\|^2_{\mathcal{F}^2(\mathcal{E}\oplus\mathcal{E}^\vee_{\bar{x},\bar{y}})}, 
\end{alignat*}
which proves that $\overline{\Lambda_{\bar{x},\bar{y}}}-i\frac{\eta_{\bar{y}}}{2\pi}$ is a set of sampling for 
$\mathcal{F}^2(\mathcal{E}\oplus\mathcal{E}^\vee_{\bar{x},\bar{y}})$.
\end{proof}
\begin{cor}
Let $\mathcal{G}(\Psi_{(\bar{x},\bar{y})},\Lambda_{b_1,...,b_n}\oplus\Lambda_{c_1,...,c_n}^\vee)$ be the Gabor system with window function as defined in (\ref{Ewindow}) and lattice as defined in (\ref{lemLambda}).  If $0<b_i<1$ and $b_i=\pm c_i$ , then $\mathcal{G}(\Psi_{(\bar{x},\bar{y})},\Lambda_{b_1,...,b_n}\oplus\Lambda_{c_1,...,c_n}^\vee)$ satisfies the frame condition.

\end{cor}
\begin{proof}
   
        Indeed, the lattice can be written as $$\mathcal{J}(\Lambda_{b_1,...,b_n}\oplus\Lambda_{c_1,...,c_n}^\vee)=\begin{pmatrix}
            c_1 & 0 & ...\\
            0& c_2 &0 \\
            \vdots & \hdots & \vdots\\
            0 & \hdots & c_n
            
        \end{pmatrix}\begin{pmatrix}
            \mathbb{Z}+i\mathbb{Z}\\
            \vdots\\
             \mathbb{Z}+i\mathbb{Z}
        \end{pmatrix}.$$
 
\end{proof}

\section{A geometric example: hypercomplex manifolds} \label{hypCSec}

  Let $(b,p)=(b_1,...,b_n,p_1,...,p_n)$ be local coordinates on $T^*M$, and $\lambda$, the tautological 1-form on $T^*M$, locally expressed as $\lambda_{(b,~ p)}=\sum_{i=1}^n p_i db_i$. If $J$ is an almost complex structure on $M$, the twist of $\lambda$ by $J$ is the 1-form $\lambda_J$ with local expression $\lambda_J=\sum_{i,j}^n p_i J^i_j db_j$. If the manifold is equipped with more than one almost complex structures, each one introduces a new 1-form. Almost hypercomplex manifolds are an example of manifolds with more than one almost complex structure.  
\begin{defn}
A manifold $B$ of dimension $4n$, is an \textbf{almost hypercomplex manifold } if it admits a triple $\{I,J,K\}$ of almost complex structures satisfying the quaternionic identities 
$$I^2=J^2=K^2=-id\text{ and } IJ=-JI=K.$$
The three almost complex structures are locally expressed as 
$$I=\sum_{k,l}I^k_ldb_l\otimes \partial b_k\text{, }J=\sum_{k,l}J^k_ldb_l\otimes \partial b_k\text{ and }K=\sum_{k,l}K^k_ldb_l\otimes \partial b_k.$$
If the almost complex structures are integrable the manifold $B$ is a {\em hypercomplex manifold}.
\end{defn}
If $(B, (I, J, K))$ is an (almost) hypercomplex manifold, each (almost) complex structure, introduces a "twisted" canonical 1-form on the cotangent bundle $T^*B$, locally expressed as
$$\lambda_I=\sum_{k,l}p_kI^k_l db_l\text{ , }\lambda_J=\sum_{k,l}p_kJ^k_l db_l\text{ and }\lambda_K=\sum_{k,l}p_kK^k_l db_l .$$
In the following proposition we prove that the 1-forms $\lambda, \lambda_I, \lambda_J$ and $\lambda_K$ restricted on the unit cotangent bundle of a hypercomplex manifold  $B$, introduce a quadruple of contact 1-forms on $\bS(T^*B)$, $$
\alpha := \lambda_{\mid_{\cS(T^*b)}},~ \alpha_I:= \lambda_I{_{\mid_{\cS(T^*b)}}}, ~ \alpha_J:= \lambda_J{_{\mid_{\cS(T^*b)}}},  ~ \alpha_K:= \lambda_K{_{\mid_{\cS(T^*b)}}},$$
which are appropriate for the construction of the Gabor system \begin{equation}
\cG(\Psi,\Lambda)=\{ M_{\lambda'} T_{\lambda} \Psi\,\mid\, (\lambda, \lambda') \in \Gamma(M,\Lambda) \} \, ,
\end{equation}  from definition \ref{defGaborE}.
\begin{lem}\label{HyperForms}
Let $M^7=\mathbb{S}(T^*B)$ be the co-sphere bundle of a hypercomplex $(B,I,J,K)$ with Riemannian metric $g$.
\begin{enumerate}
    \item The 1-forms $\alpha_I,\alpha_J$ and $\alpha_K$, induced by the twisted 1-forms $\lambda_I,\lambda_J$ and $\lambda_K$ respectively, are contact 1-forms on $M$ and form a linearly independent set of $T^*_{(b,[p])}M$ at each $(b,[p])\in \mathbb{S}(T^*B)$.
    \item  Additionally, the fibers $\pi^{-1}(b)$ for each $ b\in B$ of the fiber bundle $\pi:\mathbb{S}(T^*B)\rightarrow B$ are Legendrian submanifolds of the contact distributions induced by $\alpha_I, \alpha_J$ and $\alpha_k$ and of the canonical contact distribution induced by $\lambda\mid_{\mathbb{S}(T^*B)}$.
    \item 
If $V\rightarrow M$ is the vertical tangent bundle of $
    \pi:M\rightarrow B$, $\Gamma(V)$ denotes the vertical vector fields of $M$ and $\mathfrak{R}_{\alpha}, \mathfrak{R}_{\alpha_I}, \mathfrak{R}_{\alpha_J}, \mathfrak{R}_{\alpha_K}$ are the Reeb fields of $ \alpha, \alpha_I, \alpha_J, \alpha_K$ respectively, then $$span\{\mathfrak{R}_{\alpha}, \mathfrak{R}_{\alpha_I}, \mathfrak{R}_{\alpha_J}, \mathfrak{R}_{\alpha_K}\}\cap \Gamma(V)=
    \{0\}.$$ 
\end{enumerate}

\end{lem}
\begin{proof}
First, we are going to prove that the one forms $\alpha_I, \alpha_J, \alpha_K$ are contact.  The almost complex structure $I$ on $B$ comes from a complex structure on $B$, since it is integrable. Hence we can consider a holomorphic local coordinate chart $(U,x_1+iy_1,x_2+iy_2)$ and the corresponding local trivialization chart $(T^*U,(x_1+iy_1, x_2+iy_2, p_1^xdx_1+p_1^ydy_1+ p_2^xdx_2+p_2^ydy_2))$ of $T^*B$. The 1- form $\lambda_I$ has a local expression $\lambda_I= \sum_{i=1}^n (-p_i^ydx_i+p_i^xdy_i)$. Without loss of generality we assume that $p_1^x\neq 0$ to obtain the local expression of $\lambda_I$ restricted on $\mathbb{S}(T^*B)$,
$$\alpha_I=\lambda_{I}\mid_{\bS(T^*B)}=dy_1 -p_1^y dx_1+\sum_{i=2}^4 (-p_i ^y dx_i+p_i^x dy_i). $$
It follows that $\alpha_I\wedge (d\alpha_I)^2=\bigwedge_{i=1}^2(dp_i^x\wedge dp_i^y\wedge dx_i 
\wedge dy_i)\neq 0$.
By expressing $\lambda_J$ and $\lambda_K$ in the holomorphic coordinates induced by $J$ and $K$ respectively, one can see that $\alpha_J$ and $\alpha_K$ are contact 1-forms.
\par Next, we are going to prove that $\alpha,\alpha_I, \alpha_J$ and $\alpha_K$ are linearly independent at each $T_{(b,[p])}^*M$. Let $(b,[p])\in \bS(T^*B)$ and $\lambda, \lambda_I, \lambda_J$ and $\lambda_K$ real functions on $\bS(T^*B)$ such that 
$
    \lambda \alpha_{(b,[p])}+{\lambda_I \alpha_I}_{(b,[p])}+{\lambda_J \alpha_J}_{(b,[p])} +{\lambda_K \alpha_K}_{(b,[p])}=0
$, namely
\begin{equation}\label{lininde}
    \sum_{i=1}^4 ( p_i\lambda Id+p_i\lambda_I I^i_j + p_i\lambda_J J^i_j+ p_i\lambda_K K^i_j ) =0\, , \ \forall j\in \{1,...,4\}\, ,
\end{equation}
where $p=(p_1,..,p_4) \in T_b^*B$ is a representative of the class $[p]\in T_b{\bS(T^*B)}$.
Equation \eqref{lininde} is equivalent to $p$ being in the kernel of the linear transformation $
 (\lambda_{(b,[p])}Id+{\lambda_I}_{(b,[p])}I+{\lambda_J}_{(b,[p])}J+ {\lambda_K}_{(b,[p])}K). $

Since $p\neq 0$, the latter holds if and only if $-\lambda$ is an eigenvalue for ${\lambda_I}_{(b,[p])}I+{\lambda_J}_{(b,[p])}J+ {\lambda_K}_{(b,[p])}K$ and $p$ a corresponding eigenvector, in which case we have the following  
\begin{align*}
    ({\lambda_I}_{(b,[p])}I+{\lambda_J}_{(b,[p])}J+ {\lambda_K}_{(b,[p])}K) p&= -\lambda p \\
    \iff ({\lambda_I}_{(b,[p])}I+{\lambda_J}_{(b,[p])}J+ {\lambda_K}_{(b,[p])}K)^2 p &= \lambda^2 p\\
    \iff 
    (-\lambda_I^2 -\lambda_J^2 -\lambda_K^2)p&=\lambda ^2 p,
\end{align*}
and therefore $\lambda=\lambda_I=\lambda _J=\lambda_K=0.$ 
\par 
For each $b\in B$, the fiber $\pi^{-1}(b)$ is Legendrian with respect to any of the contact 1-forms $\alpha, ~ \alpha_I, ~ \alpha_J$ and $ \alpha_K$, since they do not depend on the local covector fields $d{p_i}$.\color{black} 
\end{proof}
 Lemma \ref{HyperForms} states that the contact 1-forms $\alpha, \alpha_I, \alpha_J$ and $\alpha_K$ share the same properties with the contact 1-forms ${F_i^*\alpha}_{i=0}^3$ introduced in section \ref{FAalpha}. Hence we obtain the following corollary which is an adaptation of \ref{lemLambda} for a hypercomplex manifold $(B,(I,J,K))$. 
\begin{cor}\label{LatticeHyper}
    The contact $1$-forms $\alpha, \alpha_I, \alpha_J$ and $\alpha_K$ on the manifold $M=\bS(T^*B)$ of a hypercomplex manifold $(B, (I,J, K))$ determine a framed lattice bundle
structure on the rank $8$ vector bundle $\cE \oplus \cE^\vee$ of \eqref{EoplusEvee}.
\end{cor}
The proof of Corollary \ref{LatticeHyper} is the same as the proof of Lemma \ref{lemLambda} . Finally, the following statement follows directly from Lemma \ref{HyperForms} and Corollary \ref{LatticeHyper}.
 \begin{prop}\label{GaborsysHyper}
Let $(B,(I,J,K)$ be a $4$-dimensional compact hypercomplex manifold. The collection
$\alpha, \alpha_I, \alpha_J, \alpha_K$ of contact $1$-forms on $M=\bS(T^*B)$ 
and the function $\Phi$ of \eqref{PhiTB} determine a Gabor system 
$\cG(\Psi,\Lambda)$  on $B$, with window function $\Psi$ as in \eqref{WindowFunction}
and with lattice bundle structure $\Lambda$ on $\cE\oplus \cE^\vee$ 
as in Corollary~\ref{LatticeHyper} (with either the Reeb
lattice bundle structure of \eqref{LambdaEdef} or the dual-basis 
lattice bundle structure of \eqref{LambdaEdef2}).
\end{prop}

\section{Applications: configuration spaces in robotics} \label{ApplSec}

In this section we describe an application of the construction described in section \ref{HigherSignal} in configuration spaces of robotic motion with environmental constraints. Consider a mechanism whose possible movements in the ambient $3$-dimensional
space are parameterized by a configuration space $\cM(R)$, which is a manifold of some higher dimension $N=\dim \cM(R)$. The configuration space represents all the possible positions the mechanisms can take, dictated by its structure. However the motion of the mechanism can be limited by environmental constraints/ obstacles which can be described as subsets of the configuration manifold. More generally, we think of a constraint according to the following definition. 

\begin{defn}
  Let $\cL_{\cM (\cR)}$ be the Lebesgue $\sigma -$algebra of $\cM (\cR)$ and $g$ a Riemannian metric on $\cM (\cR)$. A constraint on the configuration space is a probability measure 
   $$\mu: \cL_{\cM (\cR)}\rightarrow [0,1]$$
   which is absolutely continuous with respect to the volume measure $dvol_g$. 
\end{defn}
This definition allows us to consider the case of soft constraints as well. By soft constraints we mean giving a degree of preference for certain motions and
configurations over others, for example for the purpose of motion planning, instead of realizing rigid constraints which
exclude parts of the configuration space. 
A constraint on the configuration space can be identified with its density function. For simplicity we will use the letter $\mu$ for the density function $\mu: \cM(\cR)\rightarrow \bR_+$  of a constraint $\mu$.

\begin{prop}
Let $\cM_\ell(\cR)$ be the configuration space of a robotic  arm $\cR$ and $U\subset \cM_\ell (\cR)$ a relatively compact subset of configurations that are not attainable and its boundary $\partial U= \Sigma$ is smooth. For fixed $b\in \Sigma\subset B$ and for $\mu: \cM_{\ell}(\cR)\rightarrow \bR_+ $ the density function of the constraint indicating the constraints,
the output function $\cO_b(\mu,\eta_p)$ has a local maximum
for $p\in \bS^{n-1}$ the normal vector $\nu_b(\Sigma)$ at $b$ to the boundary 
hypersurface $\Sigma=\partial U$,
$$ {\rm argmax}^{loc}_{p\in \bS^{n-1}} \cO_b(\mu, \eta_p) = \nu_b(\Sigma) \, . $$  
\end{prop}

\begin{proof}
The output function $\mathcal{O}(\eta_p)$ has a local maximum when the gradient with respect to the metric ${g^*_F}_{(b,p)}$ on the fibers of $\cE$, $\nabla^{g^*_{(b,p)}}\mathcal{O}(\eta_p)$, is zero. 
Since $\mu$ is compactly supported, the lifted signal $\cI(\mu): \cE\rightarrow \bR_+$ is non-zero only on $h^{-1}(TU)$, then it follows that  

\begin{align*}
    \nabla^{g^*_{(b,p)}} \cO_b(\mathcal{I}(f), \eta_p)= \int_{h^{-1}(T_b U)} I(f)|_{h^{-1}(T_b U)}(V) i(\nabla^{g^*_{(b,p)}} \langle\eta_p,V\rangle_b) \Psi_{(b,p)}(V) \, \, dV.
\end{align*}
Thus, the gradient $\nabla^{g^*_{(b,p)}} \cO_b(\mathcal{I}(f), \eta_p)$ is equal to zero exactly when $$\nabla^{g^*_{(b,p)}}\langle\eta_p,V\rangle_b=0, \text{ for every } V\in h^{-1}(T_b U),$$
which holds exactly when $\langle\eta_p,V\rangle_b=0$ for every $V\in h^{-1}(T_b U)$, hence $\eta_p$ is normal to $\Sigma$, since $V\in T_b B$ at $b$.
\end{proof}

\begin{rem} {\rm 
If the constraint is sharp, namely $\mu=\chi_{U}$ the previous proposition follows directly from Proposition \ref{nBoundary}. }
\end{rem}

\begin{ex}[Planar Robotic Arm] {\rm 
Consider $\cR$ to be a robot arm, that is a collection of n-bars of given lengths $\ell_1,...,\ell_n $ connected with each other with joints that allow full rotation and we will denote the space of configurations of $\cR$ by $\cM_{\ell}(\cR)$.  The simplest restriction that one can impose to the motion of the robotic arm is to assume that the first joint is attached to the origin. If we allow the arm to intersect itself, then the configuration space is the n-torus,
\begin{align*}
    \cM_{\ell}(\cR)= \bS^1\times...\times \bS^1 \subset \bC^n .
\end{align*}
with coordinate charts given by the angle functions $\theta_i\mapsto e^{i\theta_i}, 0<\theta_i<2\pi $ for $i=1,..,n$.
The workspace of $\cR$ is the variety of positions on the end-point of the arm denoted as $W$ and the robot arm workspace map is the map 
\begin{align*}
  f_\cR: \cM_{\ell}(\cR)&\rightarrow W\\
  (\theta_1,..,\theta_n)&\mapsto \ell_1 \theta_1+...+\ell_n\theta_n .
\end{align*}
Suppose that the planar arm has two bars of the same length $\ell_1=\ell_2$, then the preimage $f_\cR^{-1}(0)$ is the anti-diagonal $\Delta^*=\{(\theta_1,\theta_2)\in \cM_\ell (\cR): \theta_1=-\theta_2\}$. If we consider $U$ to be $\cM_{\ell}(\cR)/\Delta^*$ and $\mu=\chi_U$, then it follows by Theorem \ref{nBoundary} that the filters 
\begin{equation*}
    \Psi_{(b,p)}(V)=exp\big(-V^tA_b V-i \langle \eta_p,V \rangle_b \big) \, , ~b\in \cM_\ell (\cR) \text{ and } p\in \bS(T^*_b \cM_\ell (\cR) )   ,
\end{equation*}
detect the configurations that make the robotic arm to fold in half.  }
\end{ex}

\bigskip
\bigskip

\subsection*{Acknowledgment} This work was supported by NSF grant DMS-2104330, 
and FQXi grants FQXi-RFP-1-804 and FQXI-RFP-CPW-2014. The authors thank Boris
Khesin, Yael Karshon and Giovanna Citti for helpful discussions.

\end{document}